\documentclass[acmsmall,screen]{acmart}

\usepackage{braket}
\usepackage{amsthm}
\newtheorem{lemma}{Lemma}
\newtheorem{problem}{Problem}
\newtheorem{example}{Example}
\usepackage{multirow}
\usepackage{graphicx}
\usepackage{subcaption}

\usepackage{algorithm}
\usepackage{algorithmic}

\usepackage[most]{tcolorbox}
\definecolor{bleudefrance}{rgb}{0.19, 0.55, 0.91}

\AtBeginDocument{%
  }

\setcopyright{acmlicensed}
\copyrightyear{2018}
\acmYear{2018}
\acmDOI{XXXXXXX.XXXXXXX}

\acmJournal{JACM}
\acmVolume{37}
\acmNumber{4}
\acmArticle{111}
\acmMonth{8}

\begin{document}

\title{StabQ: Quantum Program Analysis via Weighted Stabilizer Representations}

\author{Shangzhou Xia}
\email{xia.shangzhou.218@s.kyushu-u.ac.jp}
\orcid{0009-0006-2775-9633}
\affiliation{%
  \institution{Kyushu University}
  \city{Fukuoka}
  \country{Japan}
}

\author{Junjie Luo}
\email{junjie.luo@riken.jp}
\orcid{0009-0008-7821-6879}
\affiliation{%
  \institution{RIKEN Center for Quantum Computing}
  \city{Wako}
  \state{Saitama}
  \country{Japan}
}

\author{Jianjun Zhao}
\email{zhao@ait.kyushu-u.ac.jp}
\orcid{0000-0001-8083-4352}
\authornote{Corresponding author: zhao@ait.kyushu-u.ac.jp}
\affiliation{%
  \institution{Kyushu University}
  \city{Fukuoka}
  \country{Japan}
}

\renewcommand{\shortauthors}{Shangzhou Xia, Junjie Luo, Jianjun Zhao}

\begin{abstract}
Quantum program analysis remains challenging due to the exponentially large state space of quantum programs and the difficulty of precisely characterizing their execution behavior. In particular, non-Clifford operations introduce additional complexity that limits the applicability of stabilizer-based techniques. Although stabilizer representations provide compact descriptions for Clifford circuits, their limited expressiveness prevents them from directly supporting general quantum program analysis.
In this work, we propose \textbf{StabQ}, a symbolic execution framework for quantum program analysis based on stabilizer representations. StabQ extends stabilizer-based symbolic execution beyond Clifford-only programs by introducing a symbolic state representation that captures and propagates quantum state evolution while preserving execution semantics. Based on this representation, StabQ constructs a \textit{Tableau Chain} that represents the evolution of intermediate symbolic states throughout program execution and enables reusable analysis of quantum program executions. Furthermore, StabQ incorporates tableau consolidation and global-phase recovery mechanisms to mitigate symbolic state growth during execution.
Building upon the Tableau Chain, StabQ supports multiple quantum program analysis tasks, including quantum state reconstruction, entanglement analysis, and Clifford-property detection. We evaluate StabQ on three benchmark suites---\textit{Algorithms}, \textit{MQT Bench}, and \textit{QASMBench}. The results demonstrate that StabQ constructs semantically consistent symbolic models, accurately preserves quantum state evolution, and effectively supports downstream analysis tasks across diverse quantum programs.
\end{abstract}




\maketitle

\section{Introduction}
\label{intro}
In recent years, quantum computing has shown the potential for significant computational advantages across a range of research domains. This progress has motivated the development of general-purpose quantum computing systems. Meanwhile, the emergence of diverse quantum software stacks~\cite{qiskit,Svore_2018} has made the implementation and execution of quantum programs increasingly accessible. However, the ability to execute a quantum program does not necessarily imply the ability to understand its behavior. Characterizing how quantum states evolve during execution and identifying the properties exhibited by quantum programs remain challenging tasks. Such characterization still largely relies on manual reasoning by experts with deep knowledge of quantum mechanics.
This difficulty stems from the intrinsic complexity of quantum programs. Due to fundamental quantum phenomena such as superposition, entanglement, and probabilistic measurement outcomes, the execution behavior of quantum programs is often difficult to interpret and analyze. Moreover, these effects accumulate as computation progresses, making it increasingly challenging to understand program semantics and verify program behavior. The lack of systematic analysis techniques limits the reliability, maintainability, and trustworthiness of quantum software, creating a major obstacle to the adoption of quantum computing as a general-purpose computing paradigm.
Consequently, a dedicated framework capable of systematically modeling and analyzing quantum program execution is needed. Such a framework should provide an expressive representation of quantum states, preserve execution semantics during program evolution, and support diverse analysis tasks over quantum programs. In particular, an effective representation should balance the ability to capture complex quantum behaviors with the computational efficiency required for practical program analysis.

To analyze quantum programs, prior research has explored several approaches from different perspectives. One major direction relies on static analysis techniques, including abstract interpretation~\cite{perdrix2008entanglement,yu2021quantum}, type systems~\cite{yuan2022twist}, and bug detection frameworks~\cite{zhao2023qchecker,paltenghi2024lintq}. These approaches can efficiently reason about quantum programs and often provide soundness guarantees. However, they typically rely on abstractions that sacrifice state-level precision, which limits their ability to accurately characterize quantum execution behavior.
Another direction constructs specialized representations tailored to particular analysis objectives, such as decision diagrams~\cite{burgholzer2021advanced}, the ZX-calculus~\cite{coecke2011interacting}, and tensor networks~\cite{markov2008simulating}. These representations can achieve high precision or efficiency for specific tasks. However, they often require specialized formulations, support only a limited range of analysis objectives, and lack the flexibility needed for general-purpose quantum program analysis.
These limitations reveal a fundamental challenge in quantum software engineering: existing approaches either achieve scalability through abstraction or achieve precision through specialized representations, but few provide a unified execution model that supports accurate and efficient analysis across diverse quantum program properties.

Symbolic execution offers a promising approach to addressing this challenge. Similar to classical symbolic execution, a quantum symbolic execution framework requires an efficient symbolic representation of program states together with a propagation mechanism that incrementally updates the representation during program execution. Such a representation should faithfully preserve the semantics of quantum state evolution while remaining sufficiently compact to support scalable analysis. Furthermore, it should serve as a unified intermediate representation on which diverse program analysis tasks can be performed. Among existing quantum state representations, stabilizer representations provide a particularly attractive foundation. Instead of explicitly storing exponentially large state vectors, stabilizer representations describe quantum states using polynomial-size stabilizer tableaux, enabling efficient manipulation of a broad class of quantum states. This compact representation has demonstrated excellent scalability in stabilizer simulators for Clifford circuits~\cite{aaronson2004improved,gidney2021stim}. Moreover, stabilizer representations have been successfully employed in several analysis tasks, including Clifford-circuit equivalence checking~\cite{thanos2023fast} and entanglement analysis~\cite{honda2015analysis}. These results suggest that stabilizer representations naturally combine efficient symbolic state propagation with rich structural information, making them well suited as the underlying representation for symbolic execution. Despite these advantages, existing stabilizer representations remain fundamentally limited when applied to general quantum programs. They naturally represent only stabilizer states generated by Clifford operations, whereas practical quantum programs inevitably contain non-Clifford operations required for universal quantum computation. Once non-Clifford operations are encountered, the symbolic state can no longer be maintained using a single stabilizer representation, preventing existing approaches from serving as a general symbolic execution model. Therefore, how to extend stabilizer representations to support both Clifford and non-Clifford operations, and further build a symbolic execution framework for general quantum programs based on such extensions, remains an important research problem that has not yet been fully addressed. \textcolor{black}{Several studies~\cite{Bravyi2016,Bravyi2019} have attempted to extend stabilizer-based methods to non-stabilizer states through stabilizer decomposition techniques. These approaches typically introduce auxiliary magic states and exploit the decomposition of non-stabilizer states into linear combinations of stabilizer states, enabling the simulation of circuits composed of Clifford and non-Clifford operations, such as Clifford+T circuits. By transforming non-stabilizer components into stabilizer-state representations, they effectively handle the evolution of non-stabilizer states under Clifford operations. However, existing stabilizer decomposition approaches primarily focus on improving simulation efficiency and reducing computational overhead. To achieve scalability, many approaches rely on approximation strategies, such as truncating low-weight stabilizer components, which may introduce errors into intermediate state representations. Consequently, these methods are less suitable for precise quantum program analysis tasks that require faithful preservation of program semantics and accurate characterization of intermediate quantum states.}

In this paper, we propose \textbf{StabQ}, a symbolic execution framework for quantum program analysis based on the stabilizer representation. StabQ introduces a symbolic quantum state representation that expresses program states as weighted combinations of stabilizer tableaux. This representation enables incremental state updates during program execution while preserving the semantic information required for subsequent analysis tasks.
To extend stabilizer-based symbolic execution beyond the Clifford-only setting, StabQ introduces a Pauli-decomposition mechanism that transforms non-Clifford operations into weighted combinations of Pauli operators. Since Pauli operators admit efficient stabilizer-based transformations, this mechanism allows non-Clifford operations to be integrated into the same symbolic execution process without introducing an additional state representation. Unlike stabilizer decomposition approaches that primarily target simulation efficiency through decomposing non-stabilizer states, StabQ uses this decomposition as an execution-time transformation mechanism. Rather than optimizing the compactness of the final state representation, StabQ explicitly tracks the effects of non-Clifford operations during execution and maintains intermediate symbolic states for downstream program analysis.
Based on this symbolic representation, StabQ constructs a \textit{Tableau Node} for each program execution step. Each node records the symbolic quantum state together with the corresponding tableau components and weights. These nodes are organized into a \textit{Tableau Chain}, which captures the complete execution history of a quantum program and enables analysis over intermediate program states. To improve scalability, StabQ further incorporates tableau consolidation and global-phase recovery mechanisms. These techniques reduce redundant symbolic representations while maintaining the semantic consistency of state evolution.

Building upon the Tableau Chain, StabQ provides a unified foundation for multiple quantum program analysis tasks, including quantum state reconstruction, entanglement analysis, support-set computation, and Clifford-property detection. By combining symbolic execution with stabilizer-based representations, StabQ enables systematic analysis of quantum programs containing both Clifford and non-Clifford operations.
We evaluate StabQ through three experiments. The first evaluates the semantic correctness of the constructed Tableau Chain models by comparing reconstructed quantum states with those obtained from an exact statevector simulator. The second investigates the efficiency of downstream analysis tasks, including quantum state reconstruction and entanglement analysis. The third studies the scalability of StabQ under different circuit characteristics, including qubit number, circuit size, and Clifford rate. The results demonstrate that StabQ constructs semantically consistent symbolic execution models, supports accurate quantum program analysis, and achieves practical scalability across diverse quantum programs.

Our contributions are summarized as follows:
\begin{itemize}
    \item We introduce \textbf{StabQ}, a symbolic execution framework for quantum program analysis based on stabilizer representations. StabQ extends stabilizer-based state evolution beyond Clifford-only programs by representing quantum program states as weighted combinations of stabilizer tableaux and by supporting general quantum programs that contain non-Clifford operations.

    \item We design the \textit{Tableau Chain} as an intermediate symbolic execution representation that captures the evolution of quantum program states throughout execution. Together with tableau consolidation and global-phase recovery mechanisms, the proposed representation efficiently maintains symbolic states while preserving their semantic consistency.

    \item We integrate multiple quantum program analysis tasks, including quantum state reconstruction, entanglement analysis, support-set computation, and Clifford-property detection, on top of the Tableau Chain. Extensive experiments demonstrate that StabQ constructs semantically consistent symbolic execution models, enables effective downstream analysis, and achieves practical scalability across diverse quantum programs.
\end{itemize}

The remainder of the paper is organized as follows. Section~\ref{sec:background} offers essential background information on quantum programs and the stabilizer formalism. Section~\ref{method} elaborates on the methodology of our approach. Section~\ref{experiment} details the experimental results and the analysis. Section~\ref{discussion} provides an analysis of the current limitations of StabQ and outlines potential directions for future improvement. 
Section~\ref{validity} discusses potential threats to the validity of our study.
Section~\ref{related} reviews work related to our research. Section~\ref{conclusion} summarizes the paper and provides concluding remarks.

\section{Preliminaries}
\label{sec:background}
This section presents the essential background on quantum programs and stabilizer formalism ~\cite{nielsen2010quantum,gottesman1997}, which provides the necessary foundations for the subsequent developments in this paper.

\subsection{Quantum States, Gates and Circuits}
Quantum computation is described as the manipulation of quantum states by quantum gates organized into quantum circuits.
A quantum state is represented as a unit vector in a Hilbert space. For an $n$-qubit system, the state space is given by
\[
|\psi\rangle \in \mathcal{H}_{2^n}, \quad |\psi\rangle = \sum_{i=0}^{2^n-1} \alpha_i |i\rangle,
\]
where $\alpha_i \in \mathbb{C}$ and $\sum_i |\alpha_i|^2 = 1$.
Quantum gates are unitary operators acting on one or more qubits, and they serve as the basic building blocks of quantum computation. Each gate $U$ transforms a quantum state as
\[
|\psi'\rangle = U |\psi\rangle.
\]
A quantum circuit is an ordered sequence of quantum gates applied to an initial state, typically $|0\rangle^{\otimes n}$. The execution of a circuit is expressed as
\[
|\psi_{\text{out}}\rangle = U_m \cdots U_2 U_1 |\psi_{\text{in}}\rangle,
\]
where each $U_i$ corresponds to a quantum gate. This gate-based model captures the step-wise evolution of quantum states and forms the basis for quantum computation, simulation, and analysis.
The Pauli operators are a set of single-qubit unitary operators that are central to the analysis of quantum circuits. They are defined as
\[
I =
\begin{pmatrix}
1 & 0 \\
0 & 1
\end{pmatrix}, \quad
X =
\begin{pmatrix}
0 & 1 \\
1 & 0
\end{pmatrix}, \quad
Y =
\begin{pmatrix}
0 & -i \\
i & 0
\end{pmatrix}, \quad
Z =
\begin{pmatrix}
1 & 0 \\
0 & -1
\end{pmatrix}.
\]
On a single qubit, $X$ performs a bit flip, $Z$ performs a phase flip, and $Y$ combines the two up to a global phase. For an $n$-qubit system, an $n$-qubit Pauli operator is a tensor product of $n$ single-qubit Pauli operators, together with a global phase in $\{\pm 1, \pm i\}$. These operators form the Pauli group $\mathcal{P}_n$ under multiplication. The Pauli group underlies the stabilizer formalism, in which a quantum state is described by the Pauli operators that stabilize it.
Building on the Pauli operators, several elementary gates play a fundamental role in constructing quantum circuits. The Hadamard gate $H$, phase gate $S$, and controlled-NOT gate (CNOT) are defined as follows:
\[
H = \frac{1}{\sqrt{2}}
\begin{pmatrix}
1 & 1 \\
1 & -1
\end{pmatrix}, \quad
S =
\begin{pmatrix}
1 & 0 \\
0 & i
\end{pmatrix}, \quad
\text{CNOT} |ctl, tgt\rangle = |ctl, tgt \oplus ctl\rangle.
\]
The Hadamard gate maps computational basis states to superposition states. The phase gate $S$ applies a relative phase of $i$ to the $\lvert 1 \rangle$ state. The controlled-NOT gate is a two-qubit operation that flips the target qubit when the control qubit is in state $|1\rangle$. Together with the Pauli gates, $H$, $S$, and CNOT generate the Clifford group, and they play a central role in quantum computation.
From a computational perspective, quantum gates are categorized into Clifford and non-Clifford operations. The Clifford group consists of the unitary operators that map the Pauli group to itself under conjugation,
\begin{equation}
U P U^\dagger \in \mathcal{P}_n,
\qquad
\forall P \in \mathcal{P}_n.
\label{eq:clifford}
\end{equation}
This property keeps Clifford circuits within a structured representation that can be tracked efficiently. Non-Clifford gates fall outside this structure and extend the expressive power of quantum circuits beyond it. A commonly used example is the T gate,
\[
T =
\begin{pmatrix}
1 & 0 \\
0 & e^{i\pi/4}
\end{pmatrix},
\]
which, together with the Clifford gates, forms a universal gate set for quantum computation.
By the Gottesman--Knill theorem~\cite{GK1,aaronson2004improved,gottesman1997}, quantum circuits composed solely of Clifford gates can be simulated efficiently on classical computers, whereas the inclusion of non-Clifford gates breaks this efficient simulability and enables universal quantum computation. This distinction is fundamental to the origin of quantum computational advantage in the gate-based model. The presence and structure of non-Clifford gates indicate the computational complexity of a quantum circuit, and they are widely used in characterizing quantum resources.

\subsection{Stabilizer Formalism}
\label{stab}
The stabilizer formalism provides a compact algebraic framework for representing and manipulating an important class of quantum states known as stabilizer states. Let $\mathcal{P}_n$ denote the $n$-qubit Pauli group generated by tensor products of the single-qubit Pauli operators $I$, $X$, $Y$, and $Z$, together with the phase factors $\{\pm 1, \pm i\}$.

A stabilizer state is specified by an abelian subgroup
\begin{equation}
S \subseteq \mathcal{P}_n
\end{equation}
that contains $n$ independent generators and does not contain $-I$. The corresponding quantum state $|\psi\rangle$ is the unique simultaneous $+1$ eigenstate of all stabilizer generators,
\begin{equation}
g|\psi\rangle = |\psi\rangle,
\qquad
\forall g \in S.
\end{equation}

As an example, the Bell state
\[
|\Phi^+\rangle=\frac{|00\rangle+|11\rangle}{\sqrt{2}}
\]
is a stabilizer state. Its stabilizer group is generated by the Pauli operators
\[
S=\langle +XX,\ +ZZ\rangle,
\]
where the leading ``$+$'' denotes that the corresponding generator has eigenvalue $+1$. Throughout this paper, in a Pauli string of the form $\pm P_0P_1\cdots P_{n-1}$, the first Pauli operator $P_0$ corresponds to qubit $q_0$, followed by $P_1,\ldots,P_{n-1}$. Indeed,
\[
(+XX)|\Phi^+\rangle=|\Phi^+\rangle,\qquad
(+ZZ)|\Phi^+\rangle=|\Phi^+\rangle.
\]
The Bell state is therefore characterized uniquely as the simultaneous $+1$ eigenstate of the generators $+XX$ and $+ZZ$.

The stabilizer formalism is powerful because Clifford operations preserve the Pauli group under conjugation~\eqref{eq:clifford}. Clifford evolution can therefore be represented entirely through transformations of the stabilizer generators, which avoids explicit manipulation of exponentially large state vectors.

In practice, stabilizer states are commonly represented using stabilizer tableaux. A stabilizer tableau gives a binary representation of the stabilizer generators through the symplectic formalism. For an $n$-qubit stabilizer state, each generator can be written as
\begin{equation}
g_i = \lambda_i X^{a_i} Z^{b_i},
\end{equation}
where $\lambda_i \in \{\pm 1, \pm i\}$ denotes the phase factor, and $a_i, b_i \in \mathbb{F}_2^n$ specify the positions of the $X$ and $Z$ components. Under the binary symplectic representation, each generator is encoded by a row vector
\begin{equation}
(a_i \mid b_i \mid \lambda_i),
\end{equation}
where the first $n$ bits form the $X$ mask, the next $n$ bits form the $Z$ mask, and the phase entry stores the associated Pauli phase.

Collecting all generators yields the tableau representation
\begin{equation}
t=(X \mid Z \mid r),
\end{equation}
where $X \in \mathbb{F}_2^{n\times n}$ records the $X$ components of all generators, $Z \in \mathbb{F}_2^{n\times n}$ records the $Z$ components, and $r$ stores the corresponding phase information. As an example, the Bell state is represented as
\[
\begin{array}{c|cc|cc|c}
 & X_1 & X_2 & Z_1 & Z_2 & r\\
\hline
+XX & 1 & 1 & 0 & 0 & 0\\
+ZZ & 0 & 0 & 1 & 1 & 0
\end{array}
\]
where the first two columns encode the $X$ components, the next two columns encode the $Z$ components, and the final column records the phase of each stabilizer generator.

In practical simulators, the tableau additionally contains a set of destabilizer generators, giving a complete representation
\begin{equation}
t=(S,D),
\end{equation}
where $S$ denotes the stabilizer generators and $D$ denotes the corresponding destabilizer generators. The stabilizer component alone determines the represented state, while the destabilizers serve as auxiliary structures that support efficient tableau updates under Clifford evolution.

Stabilizer methods give highly efficient representations for Clifford circuits, but universal quantum computation also requires non-Clifford operations. Unlike Clifford operations, non-Clifford operations do not preserve the Pauli group under conjugation, so Eq.~\ref{eq:clifford} does not hold in general. The stabilizer generators are then no longer closed under evolution, and the conventional stabilizer formalism and tableau propagation techniques become inapplicable.

For example, the $T$ gate acts by conjugation as
\[
T X T^\dagger = \frac{1}{\sqrt{2}}(X+Y),
\]
which is no longer a Pauli operator. A stabilizer state therefore evolves into a linear combination of stabilizer states, which the conventional stabilizer formalism can no longer represent or propagate directly.

\section{Methodology}
\label{method}
We present the design of StabQ, a symbolic execution framework for quantum program analysis based on stabilizer representations. The methodology consists of two main components: (1) symbolic state execution and propagation through the construction of the Tableau Chain, which maintains the evolution of weighted stabilizer representations during program execution, and (2) quantum program analysis mechanisms built upon the resulting symbolic execution representation.

\subsection{Workflow}
\label{workflow}

\begin{figure}[htbp]
    \centering
    \includegraphics[width=1\linewidth]{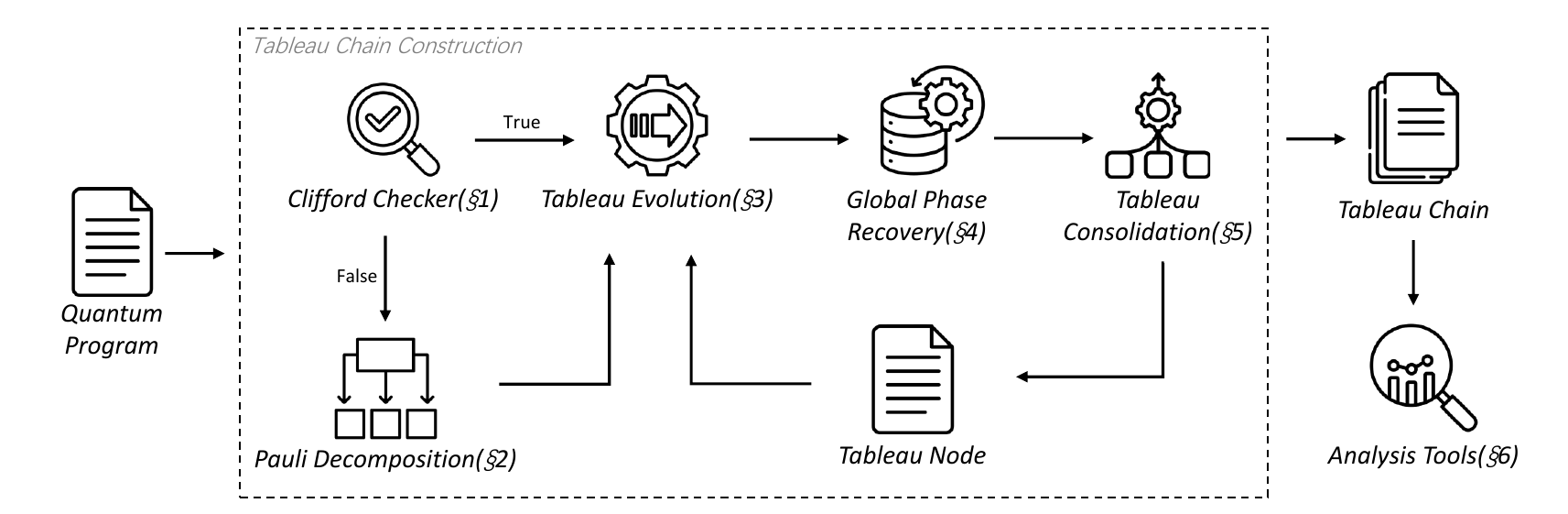}
    \caption{Workflow of StabQ.}
    \label{fig:workflow}
\end{figure}

We now describe the overall workflow of StabQ as illustrated in \autoref{fig:workflow}. Given an input quantum program, StabQ performs symbolic execution by incrementally constructing a Tableau Chain based on the stabilizer formalism. The Tableau Chain serves as an intermediate representation that records the evolution of symbolic quantum states, where each symbolic state is represented as a weighted combination of stabilizer tableaux. In addition, the representation maintains global-phase information introduced during symbolic propagation. This execution model provides a unified foundation for downstream quantum program analysis, enabling property inference and behavioral reasoning over intermediate program states.

The construction process operates at the granularity of individual quantum gates. For each gate, StabQ first applies a Clifford checker (§1) to determine whether the operation belongs to the Clifford group. Clifford operations preserve the stabilizer representation and can therefore be efficiently propagated using standard tableau update rules (§3). For non-Clifford operations, StabQ applies a Pauli-decomposition mechanism (§2), which decomposes the operation into a weighted sum of Pauli operators. These weighted Pauli components are then incorporated into the symbolic state propagation process (§3), enabling the execution model to support general quantum programs beyond the Clifford-only setting.
After each gate-level propagation step, StabQ performs global-phase recovery to maintain a consistent phase representation among symbolic states. The recovered phase information is integrated into the corresponding tableau representation (§4). Subsequently, a tableau consolidation procedure (§5) merges equivalent stabilizer components to reduce redundant symbolic states and control the growth of the Tableau Chain representation.
Each updated symbolic state is encapsulated into a Tableau Node and appended to the Tableau Chain, forming a step-wise symbolic execution representation of the input quantum program. After all gates have been processed, the complete Tableau Chain captures the entire symbolic execution history of the program.

Finally, the constructed Tableau Chain is used as a unified intermediate representation for downstream quantum program analysis tasks (§6), including quantum state reconstruction and entanglement characterization. These analysis procedures operate directly on the symbolic representation without requiring repeated quantum state reconstruction, enabling efficient analysis over the complete program execution.

\subsection{Stabilizer-Based Tableau Chain Construction}
\label{propagation}

\subsubsection{Tableau Node Representation}
\label{node}
Due to the efficiency of the stabilizer formalism in modeling quantum state evolution, our framework represents the evolution process of quantum programs through stabilizer-state propagation. Considering that quantum program analysis requires fine-grained analysis at each program step, we introduce a data structure $\mathcal{N}$ called a \textit{Tableau Node}.

\begin{equation}
    \mathcal{N}=(id, \mathcal{T}, W, Parent, Child)
\end{equation}
where $id$ denotes the program step associated with the current operation. $\mathcal{T}$ represents the list of stabilizer tableaux describing the current program state, while $W$ stores the corresponding weights associated with each tableau. $Parent$ and $Child$ are references to the Tableau Nodes corresponding to the previous and subsequent program steps, respectively, thereby linking individual nodes into the \textit{Tableau Chain} that represents the evolution of the entire quantum program.

At the beginning of program execution, we construct an initial node based on the number of qubits, denoted as $\mathcal{N}_0=(0, \mathcal{T}_0, W_0, None, None)$, where $\mathcal{T}_0$ is defined as a list containing only the stabilizer tableau $t_0$ corresponding to the initial state $\ket{00..0}$. Since the initial state consists of a single tableau component, the associated $W_0$ is initialized as a singleton list containing the value $1$, representing a weight of 1, indicating that the initial state is represented entirely by this single stabilizer tableau.
Subsequent nodes are generated incrementally through the application of quantum operations, thereby capturing the evolution of the program step by step.

The construction of each new node consists of two main stages. In the first stage, $\mathcal{T}$ and $W$ are updated according to the quantum operation at the next program step. This update procedure differs depending on whether the operation is a Clifford gate or a non-Clifford gate, and the two corresponding update strategies are detailed in section~\ref{clifford} and ~\ref{non-clifford}.
In the second stage, a global-phase recovery mechanism (section~\ref{global phase}) is applied to adjust $W$. After this correction, tableaux in $\mathcal{T}$ that correspond to the same stabilizer state are identified and merged by aggregating their associated weights (section~\ref{combination}), enabling an efficient reduction of redundant representations.

Finally, the newly constructed node updates its fields $id$, $Parent$, and $Child$, thereby forming a complete evolutionary chain that represents the execution history of the quantum program.

\subsubsection{Clifford Evolution}
\label{clifford}
The evolution of common Clifford gates, such as the Hadamard (H) gate, Pauli-X (X) gate, S gate, and CNOT gate, on a single stabilizer tableau has been extensively studied in prior work~\cite{aaronson2004improved}, and their corresponding transformation rules have been well established. For instance, applying these Clifford operations to the stabilizer representation of a Bell state yields the following evolution behavior:
\begin{equation}
\small
[+XX, +ZZ] \xrightarrow{H(0)} [+XZ, +ZX] \xrightarrow{X(0)} [-XZ, +ZX] \xrightarrow{S(0)} [-XZ, +ZY] \xrightarrow{\text{CNOT(0,1)}} [-XZ, +YX]
\end{equation}

The advantage of this representation lies in its compact and intuitive characterization of stabilizer evolution under successive Clifford operations, where the effect of each operation on the quantum state is explicitly traceable at every step. Compared with full circuit-level or state-vector representations, this approach avoids the exponential overhead associated with direct state simulation while preserving the complete semantic structure of the computation. As a result, it provides a scalable yet expressive foundation that facilitates subsequent verification and analysis tasks.

For these common Clifford gates $C$, we directly adopt the established transformation rules to define the update mechanism of $\mathcal{T}_{i}$ and $W_{i}$ within $\mathcal{N}_i=(i,\mathcal{T}_{i}, W_{i},\mathcal{N}_{i-1},\mathcal{N}_{i+1})$:
\begin{equation}
\mathcal{T}_{i-1}=[t_0,t_1,...t_k]\xrightarrow{C} \mathcal{T}_{i}=[C(t_0),C(t_1),...C(t_k)];\ \ W_{i-1} \xrightarrow{C} W_{i}=Phase(\mathcal{T}_{i-1},\mathcal{T}_{i},C, W_{i-1})
\end{equation}
In this process, since Clifford operations induce a deterministic one-to-one evolution on stabilizer tableaux, each tableau $t_{j}$ in the parent node can be evolved independently. The associated weights are initially propagated unchanged to the newly generated node.
It is worth noting that, during this transformation process, the weight $w_j \in W_{i-1}$ associated with a tableau $t_j$ can be directly propagated to its evolved tableau $C(t_j)$. This direct inheritance is valid when the quantum state is represented by a single stabilizer tableau.
However, when the state is represented as a weighted combination of multiple tableaux, the inherited weight $w_j$ may no longer be sufficient to accurately capture the evolution. In particular, global-phase differences introduced during tableau propagation can affect the relative coefficients among tableau components. Therefore, an additional global-phase recovery mechanism, implemented through the \textit{Phase} function, is required to adjust the weights accordingly. This process will be described in detail in Section ~\ref{global phase}.

In addition to these elementary Clifford gates, quantum programs may also contain composite gates composed of sequences of basic Clifford operations. Since the Clifford group is closed under composition, such composite gates remain Clifford operations and therefore also admit stabilizer tableau evolution.
However, explicit stabilizer tableau transformation rules are generally unavailable for such composite gates. Therefore, once a quantum operation $C_{comp}$ is identified as a Clifford operation, we decompose $C_{comp}$ into a sequence of elementary Clifford gates and apply the corresponding tableau evolution rules successively. By evolving each constituent Clifford gate step by step, the overall stabilizer tableau transformation induced by the composite gate can be obtained.
\begin{equation}
\begin{aligned}
&C_{comp}=C_{m}...C_{2}C_{1};\ \ C_{i} \in \{H,S,X,Y,Z,CNOT...\}\\ 
&\mathcal{T}_{i-1}=[t_0,t_1,...t_k]\xrightarrow{C_{1}} \mathcal{T}_{i}^{1}=[C_{1}(t_0),C_{1}(t_1),...C_{1}(t_k)]\\ 
&\xrightarrow{C_{2}} \mathcal{T}_{i}^{2}=[C_{2}C_{1}(t_0),C_{2}C_{1}(t_1),...C_{2}C_{1}(t_k)] \\
&\xrightarrow{...} \mathcal{T}_{i}^{m}=[C_{m}...C_{2}C_{1}(t_0),C_{m}...C_{2}C_{1}(t_1),...C_{m}...C_{2}C_{1}(t_k)]=\mathcal{T}_{i};\\ 
&W_{i-1} \xrightarrow{C_{comp}} W_{i}=Phase(\mathcal{T}_{i}^{m-1},\mathcal{T}_{i}^{m},C_{m},Phase(\mathcal{T}_{i}^{m-2},\mathcal{T}_{i}^{m-1},C_{m-1},...Phase(\mathcal{T}_{i-1},\mathcal{T}_{i}^{1},C_{1},W_{i-1})));
\end{aligned}
\end{equation}
Since all constituent operations $C_{i}$ are basic Clifford gates, their corresponding tableau transformations can be composed sequentially to obtain the overall evolution induced by the composite gate. Consequently, the stabilizer tableau evolution of $\mathcal{T}_{i}$ can be constructed through the successive application of the update rules for the individual sub-operations.
However, each sub-operation may introduce a global-phase change. Although such phase factors do not affect a single stabilizer tableau representation, they influence the coefficients associated with tableaux in our weighted linear-combination framework. Therefore, after each sub-operation, the weight vector $W$ is updated using the \textit{Phase} function, which computes and restores the corresponding phase contribution. This ensures that the accumulated weights remain consistent with the correct quantum-state evolution throughout the decomposition process.

\subsubsection{Non-Clifford Evolution}
\label{non-clifford}
Current research on the stabilizer formalism has primarily focused on the evolution of purely Clifford systems; therefore, its applicability remains limited and cannot fully capture arbitrary quantum operations. To overcome this limitation, we introduce a Pauli decomposition for non-Clifford operations, representing them as weighted combinations of Pauli operators and thereby enabling an extended stabilizer-based evolution framework. For an $n$-qubit system, we decompose a non-Clifford operation $\mathcal{G}$ as follows:
\begin{equation}
\mathcal{G}=\sum_{P\in \mathcal{P}_{n}}\alpha_{p}P;\ \ \mathcal{P}_{n}=\{I,X,Y,Z\}^{\otimes n}
\end{equation}
where $\mathcal{P}_{n}$ denotes the Pauli basis of the $n$-qubit Pauli group, and $\alpha_{p}$ represents the corresponding weight associated with each basis element in the decomposition. In this way, a non-Clifford operation $\mathcal{G}$ can be represented by two lists, $\mathbf{P}_{c}=[P_0, P_1,...P_l]$ and $\mathbf{W}_c=[\alpha_0, \alpha_1,...\alpha_l]$, corresponding to the set of Pauli components and their associated weights, respectively. Importantly, each element in $\mathbf{P}_{c}$ admits a valid stabilizer tableau transformation, allowing the evolution to be carried out within the stabilizer formalism framework.

Under this Pauli-decomposition framework, we can describe the evolution of the current node $\mathcal{N}_{i-1}$ when it encounters a non-Clifford operation $\mathcal{G}$. Since each Pauli basis element can be regarded as a basic Clifford operation (up to phase), for each tableau $t_j \in \mathcal{T}_{i-1}$ and its corresponding weight $w_j\in W_{i-1} $, we apply the Clifford evolution procedure described in Section ~\ref{clifford}:
\begin{equation}
\begin{aligned}
&t_{j}^{m}=P_{m}(t_{j});\ \ w_{j}^{m}=Phase(t_{j},\ t_{j}^{m},\ P_{m},\ w_{j}\cdot\alpha_{m});\ \ P_m\in \mathbf{P}_{c};\ \ \alpha_{m} \in \mathbf{W}_c \\ 
&\mathcal{T}_{i-1}=[t_0,t_1,...t_k]\xrightarrow{\mathcal{G}} \mathcal{T}_{i}=[t_{0}^{0},...t_{0}^{l},t_{1}^{0},...t_{1}^{l},...t_{k}^{0}...t_{k}^{l}];\\ 
&W_{i-1} \xrightarrow{\mathcal{G}} W_{i}=[w_{0}^{0},...w_{0}^{l},w_{1}^{0},...w_{1}^{l},...w_{k}^{0}...w_{k}^{l}];
\end{aligned}
\end{equation}
During this process, the newly generated node $\mathcal{N}_{i}$ may contain a large number of stabilizer tableaux in $\mathcal{T}_{i}$ along with their associated weights in $W_i$. However, redundancies may arise due to the presence of equivalent tableaux. To address this issue, we perform an equivalence-checking procedure and merge identical tableaux by aggregating their corresponding weights. A detailed description of this consolidation procedure is provided in Section ~\ref{combination}.

\subsubsection{Global Phase Recovery}
\label{global phase}
To ensure phase consistency across tableau transformations, StabQ introduces a support-driven global phase recovery mechanism. Unlike prior formulations that solve Pauli decompositions over a symplectic basis, our implementation operates directly on the amplitude support induced by stabilizer tableaux.

Given a stabilizer tableau $t\in \mathcal{T}_{i-1}$, we first extract its computational-basis support $\mathcal{B}(t)$ and compute the corresponding phase map $\Phi_t: \mathcal{B}(t) \rightarrow \{1, -1, i, -i\}$ using the internal tableau phase structure. For a transformed tableau $t' \in \mathcal{T}_{i}$, we similarly compute $\Phi_{t'}$ over $\mathcal{B}(t')$.

To account for gate-induced basis permutations and phase distortions, we construct a gate-aware consistency mapping $\mathcal{M}_C$ induced by the applied operation $C$. Specifically, $\mathcal{M}_C$ describes how basis states and phases are transformed under Clifford gates (e.g., CX, CZ, H, S, S$^\dagger$), while non-Clifford gates are handled via their Pauli decomposition.

We then define a transformed reference phase:
\[
\widetilde{\Phi}_t = \mathcal{M}_C(\Phi_t),
\]
which is aligned with the support of $t'$. The global phase correction factor is computed as:
\[
\gamma = \frac{\widetilde{\Phi}_t(b)}{\Phi_{t'}(b)}, \quad b \in \mathcal{B}(t'),
\]
where $b$ ranges over the support elements. If the ratio is consistent across all support elements, we obtain a unique global phase correction $\gamma$; otherwise, we default to $\gamma = 1$.

For Clifford operations, $\mathcal{M}_C$ is implemented explicitly as gate-dependent rewrites of bitstrings and phase flips (e.g., CZ induces conditional phase inversion, H induces basis swapping, S induces $i$-phase accumulation). For non-Clifford gates, we directly apply this procedure to each Pauli-decomposed branch.

Finally, the resulting global phase correction $\gamma$ is absorbed into the weight of the corresponding tableau component:
\[
w_i \leftarrow \gamma \cdot w_i;\ \ \ \  w_i \in W_{i-1}
\]
ensuring phase-consistent tableau evolution throughout the Tableau Chain. The entire procedure is implemented by the $Phase(\mathcal{T}_{i-1},\mathcal{T}_{i},C, W_{i-1})$ function.

\subsubsection{Tableau Consolidation}
\label{combination}
The introduction of non-Clifford operations into stabilizer-based state propagation leads to a fundamental representational challenge. While Clifford gates preserve the one-to-one evolution of stabilizer tableaux, non-Clifford gates are expressed in our framework as linear combinations of multiple stabilizer components. As a result, the number of tableaux grows rapidly during circuit propagation, significantly increasing both memory consumption and computational cost.

To address this issue, we introduce a \emph{tableau consolidation} procedure that reduces redundancy in the symbolic state representation by merging equivalent stabilizer components.
A key observation is that the stabilizer subgroup $S_j$ uniquely characterizes the underlying stabilizer state, whereas the destabilizer component $D_j$ serves purely as an auxiliary gauge choice required for maintaining a consistent Clifford update rule. It does not contribute to the physical state representation and is not observable in the stabilizer semantics of the state.
\begin{lemma}[Stabilizer Equivalence Implies State Equivalence]
Let $S_1$ and $S_2$ be two stabilizer groups acting on the same set of qubits. If $S_1 = S_2$, then the corresponding stabilizer states are identical up to a global phase, i.e.,
$|S_1\rangle \equiv e^{i\theta} |S_2\rangle.$
\end{lemma}

\textbf{Proof sketch.}
A stabilizer group S uniquely defines a quantum state $|S\rangle$ as the simultaneous +1 eigenspace of all operators in S. Since stabilizer groups are maximal abelian subgroups of the Pauli group (excluding $-I$), each valid stabilizer group corresponds to a one-dimensional subspace in the Hilbert space. Therefore, if two stabilizer groups $S_1$ and $S_2$ are identical, they define the same set of stabilizer constraints and hence the same one-dimensional subspace. It follows that their associated stabilizer states differ at most by a global phase, which is physically irrelevant. \hfill $\square$

Based on this observation, we define an equivalence relation over tableaux induced by their stabilizer structure:
\begin{equation}
t_j \sim t_k \quad \Longleftrightarrow \quad S_j = S_k
\end{equation}
All tableaux belonging to the same equivalence class represent identical stabilizer states and are therefore merged into a single representative. The associated weights are aggregated as
\begin{equation}
w_{[S]} = \sum_{t_k \in [S]} w_k
\end{equation}
where [S] denotes the equivalence class induced by stabilizer group S.

This consolidation preserves the semantic meaning of the symbolic state, while eliminating redundant representations that differ only in their auxiliary destabilizer structure. In particular, since all physically relevant phase information has been explicitly accounted for in the weight domain prior to consolidation as described in Section ~\ref{global phase}, the merging operation does not introduce any loss of state information.
From an implementation perspective, each tableau is canonicalized with respect to its stabilizer generators, and this canonical stabilizer representation is used as a hash key to group equivalent tableaux. The consolidation step is applied periodically during state propagation, ensuring that the symbolic state size remains controlled even in the presence of repeated non-Clifford expansions.

Overall, tableau consolidation provides a significant reduction in propagation overhead by exploiting the redundancy in stabilizer-based representations induced by non-Clifford decompositions, while maintaining exact consistency at the level of stabilizer state semantics.

\subsubsection{Algorithm}
\label{algorithm}

This section explains how StabQ builds the complete Tableau Chain model for a target quantum program.
Algorithm \autoref{alg1} presents the overall flow of the process.

\begin{algorithm}
	\renewcommand{\algorithmicrequire}{\textbf{Input:}}
	\renewcommand{\algorithmicensure}{\textbf{Output:}}
	\caption{: Tableau Chain Construction}
	\label{alg1}
	\begin{algorithmic}[1]
        \REQUIRE $P\gets$ target quantum program \\
        \ENSURE $\mathcal{N}_{0}\gets$ the root node of the entire Tableau Chain
        \STATE $n \gets qubits\_num(P)$ \hfill \textcolor{teal}{// Qubits Number of $P$}
        \STATE $t_0 \gets get\_tableau(0,n)$ \hfill \textcolor{teal}{// Generate initial tableau of P}
        \STATE $id \gets 0$ \hfill \textcolor{teal}{// Program step counter}
		\STATE $\mathcal{N}_{0} \gets (id, [t_{0}], [1], None, None)$ \hfill \textcolor{teal}{// Initialize root node}
        \STATE $\mathcal{N}_{c} \gets \mathcal{N}_{0}$
		\FOR{$op$ in $P$}
            \STATE $id \gets id +1$
            \STATE $\mathcal{T}, W \gets \mathcal{N}_{c}.\mathcal{T},\mathcal{N}_{c} .W$
            \IF{$Clifford\_checker(op)$} 
            \STATE \textcolor{teal}{// Decompose into a sequence of basic Clifford operations }
            \STATE $op\_list \gets basic\_Clifford(op)$
            \STATE $\mathcal{T}_{c}, W_{c} \gets \mathcal{T}, W$
            \FOR{$gate$ in $op\_list$}
            \STATE $\mathcal{T}_{new} \gets evolution(\mathcal{T}_{c},\  gate)$
            \STATE $W_{new} \gets Phase(\mathcal{T}_{c},\  \mathcal{T}_{new},\  gate,\  W_{c})$
            \STATE $\mathcal{T}_{c}, W_{c} \gets \mathcal{T}_{new}, W_{new}$
            \ENDFOR
            \ELSE
            \STATE \textcolor{teal}{// Decompose into Pauli operations with weights} 
            \STATE $op\_list,\  w \gets Pauli\_decomposition(op)$
            \STATE $\mathcal{T}_{c}, W_{c} \gets \emptyset, \emptyset$
            \FOR{$gate$ in $op\_list$}
            \STATE $w_{gate} \gets w.get[gate]$
            \STATE $\mathcal{T}_{new} \gets evolution(\mathcal{T},\  gate)$
            \STATE $W_{new} \gets Phase(\mathcal{T},\  \mathcal{T}_{new}, \ gate,\  w_{gate}\cdot W) $
            \STATE $\mathcal{T}_{c}, W_{c} \gets \mathcal{T}_{c} \cup \mathcal{T}_{new},\  W_{c} \cup W_{new}$
            \ENDFOR
            \ENDIF

            \STATE $\mathcal{T}_{c}, W_{c} \gets tableau\_consolidation(\mathcal{T}_{c}, W_{c})$
            \STATE $\mathcal{N}_{new} \gets (id, \mathcal{T}_{c}, W_{c}, \mathcal{N}_{c}, None)$
            \STATE $\mathcal{N}_{c}.Child \gets \mathcal{N}_{new}$
            \STATE $\mathcal{N}_{c} \gets \mathcal{N}_{new}$
        \ENDFOR
	\end{algorithmic}  
\end{algorithm}

Given a quantum program $P$, StabQ first determines the number of qubits and constructs the initial stabilizer tableau $t_0$, representing the computational basis state 
$|0\rangle^{\otimes n}$ (line 1,2). The tableau and its initial weight are encapsulated in the root node $\mathcal{N}_0$ (line 4), which serves as the starting point of the Tableau Chain.
For each quantum operation $op$ in $P$, StabQ first determines whether $op$ is a Clifford gate (line 9). If so, the operation is decomposed into a sequence of supported elementary Clifford gates $op\_list$ (line 11). Each elementary gate is processed sequentially by first updating the current tableaux list according to the tableau evolution rules (line 14), followed by weight phase recovery (line 15). The resulting tableaux list and weight list become the symbolic state for the next elementary gate in $op\_list$ (line 16).
Otherwise, if $op$ is a non-Clifford gate, StabQ performs a Pauli decomposition to express the operation as a weighted linear combination of Pauli operators (line 20). Since the resulting Pauli terms are independent, each Pauli operator is processed independently on the current tableaux list (lines 22–25). The corresponding weights are updated according to the Pauli coefficients and the recovered phases (line 25), and all evolved tableaux and weights are aggregated into the new symbolic state (line 26).
After processing each quantum operation $op$, StabQ performs Tableau Consolidation to merge equivalent tableaux and eliminate redundant symbolic states (line 29). The consolidated tableaux list $\mathcal{T}_{c}$ and weight list $W_{c}$ are encapsulated into a new Tableau Node, which is linked to its predecessor to extend the Tableau Chain (lines 30–32).

After all quantum operations have been processed, the root node $\mathcal{N}_{0}$, together with all linked Tableau Nodes, forms the complete Tableau Chain, which serves as the symbolic model for subsequent quantum program analysis.




\subsubsection{Illustrative Example}
\label{example}
To illustrate how the Tableau Chain is constructed and updated during program execution, we consider a simple Bell-state preparation circuit followed by two non-Clifford T gates. Although small in size, this example captures the key mechanisms of the framework, including Clifford propagation, non-Clifford decomposition, global phase recovery, and tableau consolidation. 
This 2-qubit circuit is defined as
\[
H_0 \rightarrow \mathrm{CX}_{0,1} \rightarrow T_0 \rightarrow T_0.
\]
The evolution under the first two Clifford gates is as follows:
\[
\small
\begin{aligned}
&InitialState\ \ket{00}:\mathcal{T}_{0}=Stabilizer:['+IZ','+ZI']; Destabilizer:['+IX','+XI'] \\
&\xrightarrow{H_{0}} \mathcal{T}_{1}=[Stabilizer:['+IX','+ZI'];Destabilizer:['+IX','+XI']]\ \  (GlobalPhase_{1} = 1)\\ 
&\xrightarrow{CX_{0,1}} \mathcal{T}_{2}=[Stabilizer:['+XX','+ZZ'];Destabilizer:['+IZ','+XI']]\ \ (GlobalPhase_{2} =1) \\
&W_{0} = [1] \xrightarrow{H_{0}} W_{1} = [1*GlobalPhase_{1}]=[1] \xrightarrow{CX_{0,1}} W_{2}=[1*GlobalPhase_{2}]=[1]
\end{aligned}
\]
Upon encountering the non-Clifford T gate, StabQ applies Pauli decomposition to express the operation as a weighted sum of Pauli operators, thereby enabling its integration into the Tableau Chain propagation framework:
\[
T = [['I',\ 'Z'],\ \ \text{coeffs}=[0.85355339+0.35355339i,\ \ 0.14644661-0.35355339i]]
\]
where $r_{I}=0.85355339+0.35355339i$ and $r_{Z}=0.14644661-0.35355339i$ denote the coefficients associated with the Pauli operators $I$ and $Z$, respectively. Therefore, the subsequent non-Clifford propagation proceeds as follows:
\[
\small
\begin{aligned}
&Stabilizer:['+XX','+ZZ'];Destabilizer:['+IZ','+XI'] \\
&\xrightarrow{(I, \ r_{I})} t_0 = Stabilizer:['+XX','+ZZ'];Destabilizer:['+IZ','+XI']  \ \ (GlobalPhase_{3}^{1} = 1) \\
&\xrightarrow{(Z, \ r_{Z})} t_1=Stabilizer:['-XX','+ZZ'];Destabilizer:['+IZ','+XI']  \ \ (GlobalPhase_{3}^{2} = 1) \\
&\mathcal{T}_{3}= [t_0,t_1]\\
&W_{3} = [w_{0}*r_{I}*GlobalPhase_{3}^{1},\ \  w_{0}*r_{Z}*GlobalPhase_{3}^{2}] = [r_{I},r_{Z}];\ \ w_{0}=1\ \  w_{0}\in W_{2}
\end{aligned}
\]
After applying a second T gate, the evolution proceeds as follows:
\[
\small
\begin{aligned}
&t_0 =Stabilizer:['+XX','+ZZ'];Destabilizer:['+IZ','+XI'] \\
&\xrightarrow{(I, \ r_{I})} t'_{0}=Stabilizer:['+XX','+ZZ'];Destabilizer:['+IZ','+XI']  \ \ (GlobalPhase_{4}^{1} = 1) \\
&\xrightarrow{(Z, \ r_{Z})} t''_{0}=Stabilizer:['-XX','+ZZ'];Destabilizer:['+IZ','+XI']  \ \ (GlobalPhase_{4}^{2} = 1) \\
&--------------------------------------\\
&t_1=Stabilizer:['-XX','+ZZ'];Destabilizer:['+IZ','+XI'] \\
&\xrightarrow{(I, \ r_{I})} t'_{1}=Stabilizer:['-XX','+ZZ'];Destabilizer:['+IZ','+XI']  \ \ (GlobalPhase_{4}^{3} = 1) \\
&\xrightarrow{(Z, \ r_{Z})} t''_{1}=Stabilizer:['+XX','+ZZ'];Destabilizer:['+IZ','+XI']  \ \ (GlobalPhase_{4}^{4} = 1) \\
&--------------------------------------\\
&\mathcal{T}_4 = [t'_{0},t''_{0},t'_{1},t''_{1}] \\
&W_{4} = [w_{0}*r_{I}*GlobalPhase_{4}^{1},\ \  w_{0}*r_{Z}*GlobalPhase_{4}^{2},\ w_{1}*r_{I}*GlobalPhase_{4}^{3}, \\
&\ \ \ \ \ \ \ \ w_{1}*r_{Z}*GlobalPhase_{4}^{4}] = [r_{I}^{2},r_{I}*r_{Z},r_{I}*r_{Z},r_{z}^2];\ \ w_{0}=r_{I};\ w_{1}=r_{Z};\ w_{0},w_{1}\in W_{3}
\end{aligned}
\]
At this point, we observe that the four generated stabilizer sets contain identical stabilizer components. StabQ therefore performs a tableau consolidation procedure, yielding a simplified tableau representation as follows:
\[
\small
\begin{aligned}
&\mathcal{T}_{4}= [[Stabilizer:['+XX','+ZZ'];Destabilizer:['+IZ','+XI']], \\
&\ \ \ \ \ \ \ \ [Stabilizer:['-XX','+ZZ'];Destabilizer:['+IZ','+XI']]] \\
&W_{4} = [r_{I}^{2}+r_{Z}^{2},\ 2*r_{I}*r_{Z}]
\end{aligned}
\]
This completes the construction of the Tableau Chain for the entire program. The resulting representation enables static analysis at any execution step in a modular manner. The detailed analysis procedures are presented in the following section.

\subsection{Tableau-Based Program Analysis Framework}
\label{analysis method}
After constructing the complete program evolution within the stabilizer formalism, we perform a systematic, step-wise analysis over the entire execution trace. Rather than treating the quantum program as a monolithic object, our framework operates directly on the intermediate representations maintained in the Tableau Chain, where each node encodes the symbolic quantum state together with its evolution history under both Clifford and non-Clifford transformations.

At a high level, the analysis is organized into three complementary components. First, we characterize the operational structure of the program by identifying the Clifford or non-Clifford nature of each quantum gate, which determines its corresponding evolution rule at the tableau level. Second, we reconstruct and track the symbolic quantum state throughout execution, where each program state is represented as a weighted ensemble of stabilizer tableaux, enabling consistent state interpretation under phase-aware evolution. Third, we analyze entanglement structure across qubits by exploiting the stabilizer dependencies encoded within each tableau representation, allowing efficient extraction of correlation patterns without resorting to full state-vector simulation.

By leveraging the unified tableau representation maintained in the Tableau Chain, all three types of analyses can be performed incrementally at each program step. This provides a fine-grained and semantically rich view of program evolution, and serves as the foundation for the detailed discussions in the following subsections, where we elaborate on Clifford characterization, state reconstruction, and entanglement analysis respectively.



\subsubsection{Clifford Property Analysis}
\label{rate}
To distinguish between Clifford and non-Clifford operations, we begin from the observation that these two classes of quantum gates exhibit fundamentally different computational behaviors within the stabilizer formalism. Clifford operations preserve the Pauli group under conjugation and admit efficient classical simulation, whereas non-Clifford operations are essential for universal quantum computation and introduce non-trivial computational complexity. Therefore, identifying whether a given operation is Clifford or non-Clifford is a prerequisite for understanding the computational structure of a quantum program.

To this end, we adopt an operational verification procedure based on the conjugation action on the Pauli group. For an $n$-qubit system, we construct a generating set of the Pauli group consisting of single-qubit $X_i$ and $Z_i$ operators embedded into the full Hilbert space via tensor products with identities on all other qubits, yielding
$G=\{X_1,Z_1,\dots,X_n,Z_n\}$
Given a unitary operator $U$, we compute its conjugation action:
\begin{equation}
P' = UPU^\dagger,\qquad \forall P\in G
\end{equation}
and verify whether each resulting operator remains within the Pauli group. In implementation, this is achieved by decomposing $P'$ into the Pauli basis and checking whether exactly one Pauli term has non-zero weight (up to a global phase and numerical tolerance). An operation is classified as a Clifford operation if and only if all generators are mapped to Pauli operators under conjugation; otherwise, it is identified as a non-Clifford operation. This criterion is sufficient because the Pauli group is generated by single-qubit $X$ and $Z$ operators, while conjugation preserves the group multiplication structure.

Beyond serving as a structural classification, the distinction between Clifford and non-Clifford operations plays a fundamental role in quantum program analysis. Clifford operations induce deterministic stabilizer tableau updates and remain efficiently simulable within the stabilizer formalism. In contrast, non-Clifford operations trigger Pauli-based decompositions that transform a single tableau representation into a weighted ensemble of multiple tableaux, thereby increasing symbolic-state complexity and introducing branching during program evolution. Consequently, the occurrence and distribution of non-Clifford operations provide a direct indication of the non-classical computational resources utilized by a quantum program and help identify regions where classical simulability begins to deteriorate.

During the construction of the Tableau Chain, this analysis enables the framework to effectively distinguish between Clifford and non-Clifford operations and select the corresponding evolution mechanism accordingly. Clifford operations are propagated through deterministic tableau evolution rules, whereas non-Clifford operations invoke the weighted decomposition and propagation procedure described in the previous section. As a result, Clifford-property analysis serves as a fundamental decision point throughout the symbolic evolution process and provides the basis for subsequent procedures such as state reconstruction, global phase recovery, and tableau consolidation.

Furthermore, we compute the proportion of Clifford operations appearing in the quantum program and report it as an analysis metric, referred to as the \emph{Clifford Rate}:
\begin{equation}
\mathrm{CR}(Q)=\frac{N_{\mathrm{Clifford}}}{N_{\mathrm{Total}}}
\end{equation}
where $N_{\mathrm{Clifford}}$ denotes the number of Clifford operations in the program and $N_{\mathrm{Total}}$ denotes the total number of quantum operations. Since Clifford operations admit efficient stabilizer-based simulation, whereas non-Clifford operations are the primary source of quantum computational complexity, the Clifford Rate provides an informative characterization of the program's simulability and its dependence on non-classical computational resources. A higher Clifford Rate generally indicates that a larger portion of the program remains within the efficiently simulable stabilizer regime, whereas a lower ratio suggests stronger reliance on non-Clifford resources and potentially higher symbolic-analysis complexity.

Overall, this analysis not only guides the selection of appropriate evolution rules during Tableau Chain construction but also provides a quantitative measure of the computational structure and complexity of the quantum program.




\subsubsection{Support Analysis and Quantum State Reconstruction}
\label{state}
While the Tableau Chain provides an efficient symbolic representation for propagating quantum programs under both Clifford and non-Clifford operations, it primarily captures structural evolution at the level of stabilizer generators rather than explicit quantum state semantics. However, many analysis tasks require access to state-level information that is not directly encoded in the tableau representation.
In particular, support analysis is essential for understanding the distribution of computational basis states that contribute to the quantum state, which is crucial for reasoning about sparsity, interference structure, and measurement outcomes. Without explicit support information, the symbolic representation remains disconnected from observable behavior in the computational basis.
Furthermore, state reconstruction is necessary to bridge the gap between tableau-level symbolic evolution and amplitude-level quantum semantics. Since non-Clifford operations introduce weighted linear combinations of tableaux, the resulting representation cannot be interpreted as a single stabilizer state. Therefore, reconstructing the full quantum state from tableau ensembles becomes essential for capturing interference effects and for enabling downstream analyses such as entanglement characterization and probability estimation.

We briefly derive the linear constraint formulation used for support extraction from stabilizer tableaux. To simplify the presentation, we first consider the support reconstruction problem for a single tableau. The extension to a weighted collection of tableaux maintained in a Tableau Chain will be discussed subsequently.
Let $t=(S,D)$ be a stabilizer tableau with stabilizer generators $S=\{g_1,\dots,g_m\}$. Each generator $g_i$ is an $n$-qubit Pauli operator and can be written (up to a global phase) in the standard form
\begin{equation}
g_i = i^{\kappa_i} X^{a_i} Z^{b_i}
\end{equation}
where $a_i,b_i \in \mathbb{F}_2^n$ are binary vectors indicating the support of the $X$ and $Z$ components, respectively, and $\kappa_i \in \{0,1,2,3\}$ specifies the phase factor of the Pauli operator.

Consider a computational basis state $|y\rangle$, where $y \in \mathbb{F}_2^n$. The action of a Pauli operator on the $j$-th qubit satisfies
\begin{equation}
Z_j |y\rangle = (-1)^{y_j} |y\rangle, \qquad X_j |y\rangle = |y \oplus e_j\rangle, \qquad g_i |y\rangle
= i^{\kappa_i} (-1)^{b_i \cdot y} |y \oplus a_i\rangle.
\end{equation}

For $|y\rangle$ to belong to the support of the stabilizer state, it must satisfy the stabilizer constraint
\begin{equation}
g_i |y\rangle = |y\rangle, \quad \forall i
\end{equation}

Therefore, for a basis state $|y\rangle$ to contribute to a stabilizer-invariant superposition, it must be mapped by the $X^{a_i}$ component onto another basis state that remains within the support. This induces a consistency condition that the support is closed under the bit-flip structure defined by $a_i$, i.e., basis states appear in cosets of the subspace generated by $\{a_i\}$.

After restricting the analysis to a fixed coset representative (i.e., quotienting out the $X$-induced orbit structure), the bit-flip contribution is absorbed into the support structure, and only the phase constraint remains. In this reduced representation, the stabilizer condition becomes
\begin{equation}
i^{\kappa_i} (-1)^{b_i \cdot y} = (-1)^{\lambda_i},
\end{equation}
where $\lambda_i \in \mathbb{F}_2$ encodes the effective sign of the stabilizer generator.

Discarding the global $i^{\kappa_i}$ phase (which can be absorbed into generator normalization), we obtain the phase consistency condition
\begin{equation}
(-1)^{b_i \cdot y} = (-1)^{\lambda_i},
\end{equation}
which yields the linear constraint
\begin{equation}
b_i \cdot y = \lambda_i \pmod 2.
\end{equation}
Stacking all stabilizer generators yields a linear system:
\begin{equation}
Zy = r \pmod 2
\end{equation}
where $Z \in \mathbb{F}_2^{m \times n}$ is formed by stacking vectors $b_i^\top$, and $r \in \mathbb{F}_2^m$ collects the phase bits $\lambda_i$.

Therefore, the computational basis support of the stabilizer state is characterized as the solution space of this linear system. In practice, we apply Gaussian elimination over $\mathbb{F}_2$ to this linear system, which effectively eliminates the contributions induced by the $X$ components and reduces the problem to solving for the affine solution space of a single stabilizer tableau. This procedure yields the complete computational basis support associated with the given stabilizer tableau.
The solution set defines the support of the stabilizer state:
\begin{equation}
\mathcal{S}_T = \{ y \in \mathbb{F}_2^n \mid Zy = r \}
\end{equation}
which forms an affine subspace over $\mathbb{F}_2$.

A key property of stabilizer states is that all computational basis states in the support have equal amplitude magnitude. Hence, the state can be expressed as
\begin{equation}
|\psi_T\rangle = \frac{1}{\sqrt{|\mathcal{S}_T|}} \sum_{y \in \mathcal{S}_T} \omega(y)\, |y\rangle,
\end{equation}
where $\omega(y) \in \{\pm1,\pm i\}$ denotes the relative phase determined by stabilizer constraints.

To recover $\omega(y)$, we fix a reference element $y_0 \in \mathcal{S}_T$ and set $\omega(y_0)=1$, while the remaining phases are determined by consistency conditions induced by the stabilizer generators. Specifically, $X$- and $Y$-type components introduce sign flips and imaginary phase contributions under symplectic propagation, allowing a consistent phase assignment over the entire support.

For a node containing multiple tableaux $\mathcal{T}=\{t_i\}$ with associated weights $W=\{w_i\}$, each tableau is independently reconstructed as
\begin{equation}
|\psi_{t_i}\rangle = \frac{1}{\sqrt{|\mathcal{S}_{t_i}|}} \sum_{y \in \mathcal{S}_{t_i}} \omega_i(y)\, |y\rangle
\end{equation}
and the overall quantum state is obtained via linear combination:
\begin{equation}
|\Psi\rangle = \sum_i w_i |\psi_{t_i}\rangle
\end{equation}
leading to the amplitude expression
\begin{equation}
A(y) = \sum_i w_i \frac{\omega_i(y)}{\sqrt{|\mathcal{S}_{t_i}|}}
\end{equation}

This reconstruction procedure provides a mapping from tableau ensembles to explicit quantum state representations, enabling a consistent interpretation of computational basis support and amplitude structure. It thereby serves as the basis for subsequent analyses within the Tableau Chain framework, including entanglement characterization and interference evaluation.

\subsubsection{Entanglement Analysis}
\label{entangle}
Given the reconstructed quantum state within the Tableau Chain, we characterize entanglement properties among qubits, which play a central role in understanding quantum correlations, information propagation, and the structural complexity of program evolution. In particular, entanglement analysis provides essential insight into how quantum information is distributed across subsystems, and serves as a key indicator of non-classical behavior in intermediate computational states.
In contrast to full state-vector simulation, our framework does not uniformly rely on explicit construction of reduced density matrices. Instead, we exploit the structure of stabilizer-based representations to compute entanglement efficiently whenever possible. In the stabilizer regime, entanglement can be inferred directly from tableau structure without density matrix reconstruction, whereas in more general settings involving tableau ensembles or non-Clifford-induced superpositions, we resort to reduced density matrix purity as a unified quantitative measure. This hybrid design allows us to balance computational efficiency with general applicability, while maintaining consistency with amplitude-level quantum semantics.

We consider a bipartition of the qubit register into subsystem $A$ and its complement $B$. To quantify entanglement, we use the purity of the reduced state on subsystem $A$, defined as $\mathrm{Tr}(\rho_A^2)$, as a unified measure of quantum correlation strength. However, the computation of this quantity does not always require explicit construction of the reduced density matrix. Instead, we exploit the structure of the tableau representation to derive more efficient evaluation strategies whenever applicable.
Our analysis distinguishes three cases depending on the representational structure of the tableau ensemble.

\vspace{0.5em}
\textit{Case 1: Single Stabilizer State.}
When the state is represented by a single stabilizer tableau, entanglement between subsystems can be evaluated directly from the stabilizer structure without explicit state-vector reconstruction. 
This follows the stabilizer entanglement formalism introduced by Fattal et al.~\cite{fattal}, which characterizes bipartite entanglement of stabilizer states through the algebraic structure of stabilizer generators.

Let $t$ denote a stabilizer tableau on $n$ qubits, and let $A \subseteq [n]$ be a subsystem of interest with complement $B = [n]\setminus A$. 
We represent the stabilizer group in its binary symplectic form $(X \mid Z)$, where each row corresponds to an independent stabilizer generator over $\mathbb{F}_2$. To analyze subsystem entanglement, we restrict each generator to subsystem $A$ by deleting the columns associated with subsystem $B$. This yields a reduced symplectic representation $(X_A \mid Z_A)$, which captures the stabilizer constraints that remain observable on subsystem $A$. 
The entanglement between $A$ and $B$ is determined by the reduction in the number of independent stabilizer constraints induced by this restriction, which quantifies the loss of locally enforceable degrees of freedom due to correlations across the bipartition. 
Let $S_A$ denote the subgroup of stabilizers whose support is entirely contained in subsystem $A$. Let $d_A=\dim(S_A)$ denote the number of independent local stabilizer generators. In our implementation, $d_A$ is computed by restricting the stabilizer representation to subsystem $A$ and applying Gaussian elimination to identify the independent generators belonging to $S_A$. The entanglement entropy is then given by
\begin{equation} 
E(A) = |A| - d_A. 
\end{equation} 

This quantity corresponds, up to local Clifford equivalence, to the number of maximally entangled Bell pairs across the bipartition and is consistent with the stabilizer entanglement entropy formalism~\cite{fattal}. Importantly, this computation operates entirely within the algebraic stabilizer representation, avoiding explicit construction of reduced density matrices and enabling efficient entanglement evaluation directly from the tableau structure.

\textit{Case 2: Product-State Tableau Mixture.}
When the state at a node is represented as a weighted ensemble of stabilizer tableaux $\mathcal{T} = \{t_i\}$ with coefficients $W=\{w_i\}$, we consider the regime in which each tableau induces a product state with respect to the bipartition $A|B$ (e.g., after local Clifford disentangling or when no intra-branch entanglement is present). In this setting, entanglement originates solely from coherent interference between different tableau branches, rather than from correlations within individual branches.
Let $A \subseteq [n]$ be a subsystem and $B = [n]\setminus A$ its complement. Each tableau $t_i$ defines a pure stabilizer state $|\psi_i\rangle$, whose subsystem restrictions are denoted by $|\psi_i^A\rangle$ and $|\psi_i^B\rangle$, obtained by projecting the stabilizer structure onto the corresponding qubit subsets and normalizing the resulting states.
To capture cross-branch interference effects, we define subsystem overlap matrices
\begin{equation}
A_{ij} = \langle \psi_i^A | \psi_j^A \rangle, \quad
B_{ij} = \langle \psi_i^B | \psi_j^B \rangle,
\end{equation}
which encode the overlap structure induced by the tableau ensemble on each subsystem.
The purity of the reduced density matrix on subsystem $A$ is computed using the swap-trick formulation
\begin{equation}
\mathrm{Tr}(\rho_A^2)
=
\sum_{i,j,k,l}
w_i \overline{w_j} w_k \overline{w_l}
\cdot
A_{jk} A_{li} B_{ij} B_{kl}.
\end{equation}
This expression arises from expanding the tensor-product structure of $\rho \otimes \rho$ and applying the swap operator on subsystem $A$, which couples cross-branch overlaps within $A$ while contracting complementary contributions on subsystem $B$.
Importantly, all quantities are computed directly from tableau-induced overlaps, avoiding explicit construction of density matrices. This provides an efficient mechanism to capture entanglement generated purely by interference between stabilizer branches in the ensemble representation.

\vspace{0.5em}

\textit{Case 3: General Tableau Ensemble.}
When the tableau representation contains non-product stabilizer structures and does not satisfy either stabilizer-reducibility (Case 1) or product-ensemble factorization (Case 2), we resort to explicit state reconstruction as a general fallback procedure.
In this regime, the full quantum state is reconstructed in the computational basis as
\begin{equation}
|\Psi\rangle = \sum_{y \in \{0,1\}^n} A(y)\, |y\rangle,
\end{equation}
where the amplitudes $A(y)$ are obtained from the tableau-chain propagation procedure, which aggregates contributions from all tableau components, including their associated weights and global phase corrections as described in the previous sections.

Given the reconstructed state, the reduced density matrix on subsystem $A$ is obtained via partial trace over subsystem $B$:
\begin{equation}
\rho_A = \mathrm{Tr}_B(|\Psi\rangle \langle \Psi|).
\end{equation}
Entanglement can then be evaluated either through the eigenvalue spectrum of $\rho_A$ or equivalently via Schmidt decomposition. As a computationally efficient scalar proxy, the purity $\mathrm{Tr}(\rho_A^2)$ is used when full spectral information is not required.
Although this procedure has exponential cost in the worst case due to full state reconstruction, it provides a general and consistent entanglement characterization when neither stabilizer-based nor product-ensemble simplifications are applicable.

\vspace{0.5em}

This three-level entanglement characterization framework enables systematic analysis across different regimes of the Tableau Chain representation, including stabilizer-dominated, product-ensemble, and general reconstruction settings. By exploiting stabilizer structure whenever available and resorting to full state reconstruction only when necessary, the proposed approach achieves a principled trade-off between computational efficiency and representational generality, while remaining consistent with amplitude-level quantum semantics induced by the tableau-chain propagation model.

\section{Experimental Evaluation}
\label{experiment}
To evaluate the effectiveness of our proposed framework, we implement StabQ in Python based on the Qiskit Library~\cite{qiskit}. All experiments were run on a MacBook Pro with an Apple M1 Pro chip and 16 GB unified memory. In the experiments, we intend to answer the following questions:

\begin{itemize}

\item \textbf{RQ1 (Model Correctness)}: How correctly does StabQ construct Tableau Chains for quantum programs?

\item \textbf{RQ2 (Program Analysis)}:
How effectively does StabQ support quantum program analysis tasks over Tableau Chains?

\item \textbf{RQ3 (Scalability)}: How does StabQ scale with respect to circuit characteristics such as qubit number, gate number, and Clifford rate?

\end{itemize}

\begin{table}[htbp]
\small
\centering
\caption{Benchmark circuits used in the evaluation of StabQ.}
\label{tab:benchmarks}
\begin{tabular}{l|p{6cm}|c|c}
\hline
 & Benchmark Circuits & Qubits Number & Gates Number \\ \hline
Algorithms & Amplitude Estimation, Deutsch–Jozsa, Grover, QAOA, QFT, QPEExact, QPEInexact, QWalk, GHZ, VQE& 2 -- 14 & 2 -- 12327 \\ \hline
MQT Bench &  GraphState, PortfolioQAOA, PortfolioVQE, QNN, QFTEntangled, RealAmpRandom, Su2-Random, TwoLocalRandom, WState & 2 -- 14 & 2 -- 532  \\ \hline
QASMBench & Adder, Bell, Fredkin, Toffoli, Teleportation, Simon, SAT, QEC, HHL, Qrng, LPN, PEA ... & 2 -- 14 & 4 -- 689 \\ \hline
\end{tabular}
\end{table}
In \textbf{RQ1}, we evaluate the correctness of the Tableau Chain constructed by StabQ. Specifically, we assess whether the proposed symbolic execution model can faithfully capture quantum program semantics by validating that both quantum state reconstruction and entanglement analysis derived from the Tableau Chain are consistent with exact quantum-state simulations.
In \textbf{RQ2}, we investigate the effectiveness of StabQ in performing downstream analysis based on the constructed Tableau Chains. \textcolor{black}{Specifically, we evaluate the runtime overhead of two representative categories of downstream analysis tasks: quantum state reconstruction and entanglement analysis. Since support-set computation can be directly derived from reconstructed quantum states, and entanglement analysis relies on purity computation over reduced states, these two categories capture the essential analysis capabilities provided by StabQ. Therefore, we report their performance across benchmark circuits with varying characteristics.}
In \textbf{RQ3}, we examine the scalability of StabQ under varying circuit characteristics. Specifically, we investigate how factors such as the number of qubits, circuit size (in terms of gate count), and Clifford rate affect the runtime and memory consumption of the framework.

For correctness evaluation, we use \textit{qiskit.quantum\_info.Statevector} as the ground truth. 
Existing stabilizer-based simulators are not suitable as correctness baselines because they are primarily designed for Clifford circuits and cannot directly handle general quantum programs containing non-Clifford operations. 
Since StabQ aims to construct an exact symbolic representation of quantum program evolution rather than an approximate simulation model, we compare the quantum states reconstructed from Tableau Chains with the exact statevectors generated by Qiskit's statevector simulator. 
This comparison directly evaluates whether StabQ faithfully preserves the semantic evolution of quantum programs.
For benchmarks, we evaluate StabQ using three benchmark suites: \textit{Algorithms}, \textit{MQT Bench}~\cite{mqtbench}, and \textit{QASMBench}~\cite{QASMBench}, as summarized in ~\autoref{tab:benchmarks}. All evaluated circuits contain fewer than 15 qubits to ensure feasibility of exact statevector-based simulation.

\subsection{Model Correctness}
\label{Model Correctness}
\begin{table}[!htbp]
\centering
\caption{Experimental results of model construction and correctness evaluation on representative benchmark quantum programs. Each benchmark family is treated as a single experimental group. We report the ranges of qubit number, gate number, Clifford rate, model construction time (seconds), peak memory usage (MB), and the maximum number of generated tableaux for all circuits whose Tableau Chains are successfully constructed within 2000 seconds.}
\label{table:model}
\begin{tabular}{c|ccccccc}
\multirow{2}{*}{\begin{tabular}[c]{@{}c@{}}Quantum\\ Circuits\end{tabular}} & \multirow{2}{*}{\begin{tabular}[c]{@{}c@{}}Qubit \& Gate\\ Number\end{tabular}} & \multirow{2}{*}{\begin{tabular}[c]{@{}c@{}}Clifford\\ Rate\end{tabular}} & \multirow{2}{*}{\begin{tabular}[c]{@{}c@{}}Build\\ Time\end{tabular}} & \multirow{2}{*}{\begin{tabular}[c]{@{}c@{}}Peak\\ Memory\end{tabular}} & \multirow{2}{*}{\begin{tabular}[c]{@{}c@{}}Tableau\\ Number\end{tabular}} & \multicolumn{2}{c}{Analysis Results} \\ \cline{7-8} 
                          &                                                                                  &                                                                          &                                                                       &                                                                        &                                                                           & State         & Entanglement         \\ \hline
AE                        & \begin{tabular}[c]{@{}c@{}}2 -- 9\\ 7 -- 77\end{tabular}                          & \begin{tabular}[c]{@{}c@{}}0.4156\\ 0.5714\end{tabular}                  & \begin{tabular}[c]{@{}c@{}}0.011\\ 967.8\end{tabular}                 & \begin{tabular}[c]{@{}c@{}}67.86\\ 123.7\end{tabular}                  & \begin{tabular}[c]{@{}c@{}}4\\ 512\end{tabular}                           &     $\checkmark$          &      $\checkmark$                \\ \hline
DJ                        & \begin{tabular}[c]{@{}c@{}}2 -- 14\\ 4 -- 40\end{tabular}                        & \begin{tabular}[c]{@{}c@{}}1.0000\\ 1.0000\end{tabular}                  & \begin{tabular}[c]{@{}c@{}}0.001\\ 46.15\end{tabular}                 & \begin{tabular}[c]{@{}c@{}}68.89\\ 81.30\end{tabular}                  & \begin{tabular}[c]{@{}c@{}}1\\ 1\end{tabular}                             &             $\checkmark$  &   $\checkmark$                   \\ \hline
Grover                    & \begin{tabular}[c]{@{}c@{}}2 -- 9\\ 2 -- 1195\end{tabular}                       & \begin{tabular}[c]{@{}c@{}}0.5178\\ 1.0000\end{tabular}                  & \begin{tabular}[c]{@{}c@{}}0.003\\ 1468\end{tabular}                  & \begin{tabular}[c]{@{}c@{}}68.90\\ 170.6\end{tabular}                  & \begin{tabular}[c]{@{}c@{}}1\\ 64\end{tabular}                            &             $\checkmark$  &    $\checkmark$                  \\ \hline
QAOA                      & \begin{tabular}[c]{@{}c@{}}3 -- 9\\ 15 -- 45\end{tabular}                        & \begin{tabular}[c]{@{}c@{}}0.2000\\ 0.2000\end{tabular}                  & \begin{tabular}[c]{@{}c@{}}0.070\\ 1253\end{tabular}                  & \begin{tabular}[c]{@{}c@{}}68.09\\ 85.50\end{tabular}                  & \begin{tabular}[c]{@{}c@{}}4\\ 256\end{tabular}                           &             $\checkmark$  &   $\checkmark$                   \\ \hline
QFT                       & \begin{tabular}[c]{@{}c@{}}2 -- 14\\ 4 -- 112\end{tabular}                       & \begin{tabular}[c]{@{}c@{}}0.1875\\ 0.7500\end{tabular}                  & \begin{tabular}[c]{@{}c@{}}0.010\\ 842.6\end{tabular}                 & \begin{tabular}[c]{@{}c@{}}68.75\\ 88.31\end{tabular}                  & \begin{tabular}[c]{@{}c@{}}1\\ 1\end{tabular}                             &             $\checkmark$  &    $\checkmark$                  \\ \hline
QPEExact                                                                    & \begin{tabular}[c]{@{}c@{}}2 -- 10\\ 4 -- 66\end{tabular}                         & \begin{tabular}[c]{@{}c@{}}0.3636\\ 1.0000\end{tabular}                  & \begin{tabular}[c]{@{}c@{}}0.005\\ 117.1\end{tabular}                 & \begin{tabular}[c]{@{}c@{}}69.28\\ 88.31\end{tabular}                  & \begin{tabular}[c]{@{}c@{}}1\\ 64\end{tabular}                            &          $\checkmark$     &   $\checkmark$                   \\ \hline
QPEInexact                                                                  & \begin{tabular}[c]{@{}c@{}}2 -- 9\\ 57\end{tabular}                               & \begin{tabular}[c]{@{}c@{}}0.3684\\ 0.7500\end{tabular}                  & \begin{tabular}[c]{@{}c@{}}0.006\\ 567.8\end{tabular}                 & \begin{tabular}[c]{@{}c@{}}68.86\\ 97.98\end{tabular}                  & \begin{tabular}[c]{@{}c@{}}2\\ 256\end{tabular}                           &        $\checkmark$       & $\checkmark$                     \\ \hline
QWalk                     & \begin{tabular}[c]{@{}c@{}}3 -- 13\\ 25 -- 12327\end{tabular}                    & \begin{tabular}[c]{@{}c@{}}0.1849\\ 0.7600\end{tabular}                  & \begin{tabular}[c]{@{}c@{}}0.109\\ 1212\end{tabular}                  & \begin{tabular}[c]{@{}c@{}}69.38\\ 258.7\end{tabular}                  & \begin{tabular}[c]{@{}c@{}}4\\ 96\end{tabular}                            &              $\checkmark$ &    $\checkmark$                  \\ \hline
GHZ                       & \begin{tabular}[c]{@{}c@{}}2 -- 14\\ 2 -- 14\end{tabular}                        & \begin{tabular}[c]{@{}c@{}}1.0000\\ 1.0000\end{tabular}                    & \begin{tabular}[c]{@{}c@{}}0.006\\ 0.030\end{tabular}                 & \begin{tabular}[c]{@{}c@{}}68.09\\ 69.34\end{tabular}                  & \begin{tabular}[c]{@{}c@{}}1\\ 1\end{tabular}                             &             $\checkmark$  &    $\checkmark$                  \\ \hline
VQE                                                                         & \begin{tabular}[c]{@{}c@{}}3 -- 12\\ 13 -- 58\end{tabular}                        & \begin{tabular}[c]{@{}c@{}}0.3333\\ 0.4737\end{tabular}                  & \begin{tabular}[c]{@{}c@{}}0.040\\ 281.5\end{tabular}                 & \begin{tabular}[c]{@{}c@{}}69.27\\ 437.6\end{tabular}                  & \begin{tabular}[c]{@{}c@{}}8\\ 4096\end{tabular}                          &        $\checkmark$       &        $\checkmark$              \\ \hline
GraphState                & \begin{tabular}[c]{@{}c@{}}3 -- 14\\ 6 -- 28\end{tabular}                        & \begin{tabular}[c]{@{}c@{}}1.0000\\ 1.0000\end{tabular}                    & \begin{tabular}[c]{@{}c@{}}0.014\\ 112.4\end{tabular}                 & \begin{tabular}[c]{@{}c@{}}69.38\\ 85.50\end{tabular}                  & \begin{tabular}[c]{@{}c@{}}1\\ 1\end{tabular}                             &              $\checkmark$ &  $\checkmark$                    \\ \hline
PQAOA                     & \begin{tabular}[c]{@{}c@{}}3 -- 12\\ 21 -- 246\end{tabular}                      & \begin{tabular}[c]{@{}c@{}}0.0000\\ 0.0000\end{tabular}                    & \begin{tabular}[c]{@{}c@{}}0.134\\ 1953\end{tabular}                  & \begin{tabular}[c]{@{}c@{}}85.50\\ 1398\end{tabular}                   & \begin{tabular}[c]{@{}c@{}}8\\ 4096\end{tabular}                          &             $\checkmark$  &   $\checkmark$                   \\ \hline
PVQE                      & \begin{tabular}[c]{@{}c@{}}3 -- 11\\ 21 -- 209\end{tabular}                      & \begin{tabular}[c]{@{}c@{}}0.4286\\ 0.7895\end{tabular}                  & \begin{tabular}[c]{@{}c@{}}0.211\\ 1457\end{tabular}                  & \begin{tabular}[c]{@{}c@{}}70.19\\ 624.1\end{tabular}                  & \begin{tabular}[c]{@{}c@{}}8\\ 2048\end{tabular}                          &             $\checkmark$  &  $\checkmark$                    \\ \hline
QNN                       & \begin{tabular}[c]{@{}c@{}}2 -- 12\\ 15 -- 455\end{tabular}                      & \begin{tabular}[c]{@{}c@{}}0.3333\\ 0.6044\end{tabular}                  & \begin{tabular}[c]{@{}c@{}}0.037\\ 1991\end{tabular}                  & \begin{tabular}[c]{@{}c@{}}69.05\\ 2148\end{tabular}                   & \begin{tabular}[c]{@{}c@{}}4\\ 4096\end{tabular}                          &             $\checkmark$  &   $\checkmark$                   \\ \hline
QFTEntangled                                                                & \begin{tabular}[c]{@{}c@{}}2 -- 9\\ 6 -- 1008\end{tabular}                        & \begin{tabular}[c]{@{}c@{}}0.3793\\ 0.8333\end{tabular}                  & \begin{tabular}[c]{@{}c@{}}0.012\\ 1667\end{tabular}                  & \begin{tabular}[c]{@{}c@{}}67.86\\ 83.50\end{tabular}                  & \begin{tabular}[c]{@{}c@{}}4\\ 512\end{tabular}                           &     $\checkmark$            &             $\checkmark$           \\ \hline
RAR                       & \begin{tabular}[c]{@{}c@{}}2 -- 12\\ 11 -- 246\end{tabular}                      & \begin{tabular}[c]{@{}c@{}}0.2727\\ 0.8049\end{tabular}                  & \begin{tabular}[c]{@{}c@{}}0.020\\ 1069\end{tabular}                  & \begin{tabular}[c]{@{}c@{}}69.03\\ 1488\end{tabular}                   & \begin{tabular}[c]{@{}c@{}}4\\ 4096\end{tabular}                          &              $\checkmark$ &  $\checkmark$                    \\ \hline
S2R                                                                         & \begin{tabular}[c]{@{}c@{}}2 -- 12\\ 11 -- 246\end{tabular}                       & \begin{tabular}[c]{@{}c@{}}0.2727\\ 0.8049\end{tabular}                  & \begin{tabular}[c]{@{}c@{}}0.040\\ 1415\end{tabular}                  & \begin{tabular}[c]{@{}c@{}}67.95\\ 1032\end{tabular}                   & \begin{tabular}[c]{@{}c@{}}4\\ 4096\end{tabular}                          &              $\checkmark$ &           $\checkmark$           \\ \hline
TLR                                                                         & \begin{tabular}[c]{@{}c@{}}2 -- 12\\ 11 -- 246\end{tabular}                       & \begin{tabular}[c]{@{}c@{}}0.2727\\ 0.8049\end{tabular}                  & \begin{tabular}[c]{@{}c@{}}0.036\\ 1061\end{tabular}                  & \begin{tabular}[c]{@{}c@{}}69.31\\ 716.4\end{tabular}                  & \begin{tabular}[c]{@{}c@{}}4\\ 4096\end{tabular}                          &              $\checkmark$ &         $\checkmark$             \\ \hline
WState                    & \begin{tabular}[c]{@{}c@{}}2 -- 14\\ 5 -- 53\end{tabular}                        & \begin{tabular}[c]{@{}c@{}}0.5094\\ 0.6000\end{tabular}                  & \begin{tabular}[c]{@{}c@{}}0.012\\ 321.4\end{tabular}                 & \begin{tabular}[c]{@{}c@{}}69.41\\ 365.8\end{tabular}                  & \begin{tabular}[c]{@{}c@{}}2\\ 8192\end{tabular}                          &             $\checkmark$  &   $\checkmark$                   \\ \hline
QASMBench                                                                   & \begin{tabular}[c]{@{}c@{}}2 -- 13\\ 4 -- 689\end{tabular}                        & \begin{tabular}[c]{@{}c@{}}0.3333\\ 1.0000\end{tabular}                  & \begin{tabular}[c]{@{}c@{}}0.013\\ 430.6\end{tabular}                 & \begin{tabular}[c]{@{}c@{}}68.64\\ 106.4\end{tabular}                  & \begin{tabular}[c]{@{}c@{}}2\\ 128\end{tabular}                           &     $\checkmark$          &                     $\checkmark$         
\end{tabular}
\end{table}

We conducted experiments on benchmark quantum programs. We set a maximum model building time of 2000 seconds, and the results are presented in ~\autoref{table:model}.

The experimental results demonstrate that the Tableau Chains constructed by StabQ correctly support the evaluated downstream analysis tasks, including quantum state reconstruction and entanglement analysis. Specifically, we reconstruct quantum statevectors from the Tableau Chains and compare them with the statevectors generated by an exact statevector simulator. Across all evaluated benchmark circuits, the reconstructed states are identical to the simulation results, indicating that StabQ faithfully preserves the semantic evolution of quantum programs during symbolic execution. Furthermore, for entanglement analysis, we validate the purity values computed from the Tableau Chain representation by comparing them with the results obtained through Schmidt decomposition over the complete density matrix. The results show that the purity values derived from StabQ are fully consistent with those obtained from the density-matrix-based approach, further demonstrating that the symbolic representation accurately preserves the entanglement information of quantum states. Compared with direct quantum execution or statevector simulation, constructing a Tableau Chain introduces additional computational overhead because StabQ explicitly maintains symbolic states and execution information at each program step. However, unlike conventional simulation approaches that only provide the final execution result, the constructed Tableau Chain captures the complete evolution history of quantum states and enables multiple downstream analyses over arbitrary intermediate program states after a single model construction process. Therefore, the construction cost represents a trade-off for obtaining a reusable symbolic execution representation that supports comprehensive quantum program analysis. To further evaluate the compactness of the proposed representation, we record the number of tableaux contained in each Tableau Node throughout the entire Tableau Chain and report the maximum value observed for each benchmark family. Although some benchmark families, such as QPEExact, Grover, and QWalk, exhibit relatively low Clifford rates, the number of tableaux remains bounded within practical limits across the evaluated benchmarks and does not exhibit exponential growth in these cases. This suggests that the tableau consolidation strategy effectively contributes to controlling the growth of symbolic representations by merging equivalent tableaux during symbolic execution.

Overall, the results show that StabQ constructs a semantically consistent symbolic execution model while maintaining a compact Tableau Chain in practice. Although circuits with intensive non-Clifford operations may still introduce additional symbolic growth, the proposed representation provides a practical foundation for systematic quantum program analysis by enabling reusable intermediate-state representations.

\begin{tcolorbox}[size=title,rightrule=1mm, leftrule=1mm, toprule=0mm, bottomrule=0mm, arc=0pt,colback=gray!5,colframe=bleudefrance!75!black,breakable]
\textbf{Answer to RQ1:}  
StabQ correctly constructs a semantically consistent Tableau Chain that faithfully captures the execution evolution of quantum programs. The quantum states reconstructed from the generated Tableau Chains exactly match those obtained from an exact statevector simulator across all evaluated benchmark circuits. Furthermore, the purity values computed from the Tableau Chain representation are identical to those obtained through Schmidt decomposition over the complete density matrix, demonstrating that StabQ preserves both state semantics and entanglement information. In addition, the number of tableaux maintained in the Tableau Chain remains bounded in practice, even for benchmark families with relatively low Clifford rates, suggesting that tableau consolidation effectively mitigates symbolic representation growth. These results demonstrate that StabQ constructs reliable and reusable symbolic execution models for downstream quantum program analysis.
\end{tcolorbox}

\subsection{Program Analysis}
\label{Analysis Efficiency}
We evaluate the capability of StabQ in supporting quantum program analysis tasks based on the constructed Tableau Chains using benchmark quantum programs. The experimental results are summarized in ~\autoref{tab:analysis}.

For the quantum state reconstruction task, StabQ performs analysis directly on the symbolic representation maintained by the Tableau Chain. Specifically, support analysis, as an essential step in the state recovery process, employs a method based on $Zy=r \pmod{2}$ to determine the computational-basis support of quantum states, and combines it with relative phase reconstruction to recover complete quantum state information. This process enables the analysis module to directly access intermediate quantum states recorded in the Tableau Chain without requiring re-execution of the original quantum program.
For the entanglement analysis task, StabQ infers the entanglement properties of intermediate quantum states based on the symbolic states stored in the Tableau Chain. Specifically, we compute the purity of reduced states and compare the results with those obtained through Schmidt decomposition over the corresponding full density matrices. The experimental results show that the purity values computed by StabQ are identical to those obtained from Schmidt decomposition, demonstrating that the Tableau Chain preserves the semantic information required for accurate entanglement inference. For Cases 1 and 3 (Section~\ref{entangle}), StabQ directly performs entanglement analysis on the evaluated states. For Case 2, although the corresponding circuits exhibit product-state structures, the current analysis procedure does not explicitly exploit such structural information, resulting in additional symbolic processing overhead during tableau manipulation.

Overall, the experimental results demonstrate that the Tableau Chain serves not only as a symbolic representation of quantum program execution but also as a unified intermediate representation supporting multiple analysis tasks. Compared with the one-time construction process of the Tableau Chain, subsequent analysis tasks incur significantly lower overhead, highlighting the reusability of the proposed representation in quantum program analysis scenarios. Unlike approaches that require re-executing programs or reconstructing states separately for different analysis tasks, StabQ enables reusable analysis over arbitrary execution steps through a single Tableau Chain construction process.

\begin{table}[!htbp]
\centering
\caption{Experimental results of quantum state analysis and entanglement analysis on the benchmark quantum programs. We perform both analyses on all circuits in Table~\ref{tab:benchmarks}, with a maximum analysis time of 500 seconds. For each benchmark family, we report the range of analysis time (seconds). For entanglement analysis, we report the entanglement analysis case (Section~\ref{entangle}) and the range of tableau number at the analyzed program step.}
\label{tab:analysis}
\begin{tabular}{c|cc|cccc}
\multirow{2}{*}{\begin{tabular}[c]{@{}c@{}}Quantum\\ Circuits\end{tabular}} & \multicolumn{2}{c|}{State}                                                                                        & \multicolumn{4}{c}{Entanglement}                                                                                                                                                   \\ \cline{2-7} 
                                                                            & \begin{tabular}[c]{@{}c@{}}Qubit \\ Number\end{tabular} & Time                                                    & \begin{tabular}[c]{@{}c@{}}Qubit\\ Number\end{tabular} & Case & \begin{tabular}[c]{@{}c@{}}Tableau\\ Number\end{tabular} & Time                                                    \\ \hline
AE                                                                          & 2 -- 9                                                  & \begin{tabular}[c]{@{}c@{}}0.0006\\ 0.1314\end{tabular} & 2 -- 7                                                 & 2    & \begin{tabular}[c]{@{}c@{}}4\\ 256\end{tabular}          & \begin{tabular}[c]{@{}c@{}}0.0012\\ 218.77\end{tabular} \\ \hline
DJ                                                                          & 2 -- 14                                                 & \begin{tabular}[c]{@{}c@{}}0.0003\\ 0.0032\end{tabular} & 2 -- 14                                                & 1    & 1                                                        & \begin{tabular}[c]{@{}c@{}}0.0001\\ 0.0001\end{tabular} \\ \hline
Grover                                                                      & 2 -- 9                                                  & \begin{tabular}[c]{@{}c@{}}0.0003\\ 0.4502\end{tabular} & 2 -- 9                                                 & 1,2  & \begin{tabular}[c]{@{}c@{}}1\\ 64\end{tabular}           & \begin{tabular}[c]{@{}c@{}}0.0001\\ 19.012\end{tabular} \\ \hline
QAOA                                                                        & 3 -- 9                                                  & \begin{tabular}[c]{@{}c@{}}0.0016\\ 14.584\end{tabular} & 2 -- 8                                                 & 2    & \begin{tabular}[c]{@{}c@{}}4\\ 128\end{tabular}          & \begin{tabular}[c]{@{}c@{}}0.0023\\ 212.30\end{tabular} \\ \hline
QFT                                                                         & 2 -- 14                                                 & \begin{tabular}[c]{@{}c@{}}0.0004\\ 3.0147\end{tabular} & 2 -- 14                                                & 1    & 1                                                        & \begin{tabular}[c]{@{}c@{}}0.0001\\ 0.0001\end{tabular} \\ \hline
QPEExact                                                                    & 2 -- 10                                                 & \begin{tabular}[c]{@{}c@{}}0.0003\\ 0.0008\end{tabular} & 2 -- 10                                                & 1    & 1                                                        & \begin{tabular}[c]{@{}c@{}}0.0001\\ 0.0001\end{tabular} \\ \hline
QPEInexact                                                                  & 2 -- 9                                                  & \begin{tabular}[c]{@{}c@{}}0.0004\\ 0.0673\end{tabular} & 2 -- 8                                                 & 2    & \begin{tabular}[c]{@{}c@{}}2\\ 128\end{tabular}          & \begin{tabular}[c]{@{}c@{}}0.0003\\ 212.99\end{tabular} \\ \hline
QWalk                                                                       & 3 -- 13                                                 & \begin{tabular}[c]{@{}c@{}}0.0004\\ 0.0059\end{tabular} & 3 -- 13                                                & 1,2  & \begin{tabular}[c]{@{}c@{}}1\\ 8\end{tabular}            & \begin{tabular}[c]{@{}c@{}}0.0001\\ 0.2566\end{tabular} \\ \hline
GHZ                                                                         & 2 -- 14                                                 & \begin{tabular}[c]{@{}c@{}}0.0004\\ 0.0031\end{tabular} & 2 -- 14                                                & 1    & 1                                                        & \begin{tabular}[c]{@{}c@{}}0.0001\\ 0.0001\end{tabular} \\ \hline
VQE                                                                         & 2 -- 12                                                 & \begin{tabular}[c]{@{}c@{}}0.0012\\ 1.2138\end{tabular} & 2 -- 12                                                & 3    & \begin{tabular}[c]{@{}c@{}}8\\ 3414\end{tabular}         & \begin{tabular}[c]{@{}c@{}}0.0014\\ 1.3932\end{tabular} \\ \hline
Graphstate                                                                  & 3 -- 14                                                 & \begin{tabular}[c]{@{}c@{}}0.0006\\ 3.0851\end{tabular} & 3 -- 14                                                & 1    & 1                                                        & \begin{tabular}[c]{@{}c@{}}0.0001\\ 0.0001\end{tabular} \\ \hline
PQAOA                                                                       & 3 -- 12                                                 & \begin{tabular}[c]{@{}c@{}}0.0011\\ 1.4299\end{tabular} & 3 -- 7                                                 & 2    & \begin{tabular}[c]{@{}c@{}}8\\ 128\end{tabular}          & \begin{tabular}[c]{@{}c@{}}0.0106\\ 211.84\end{tabular} \\ \hline
PVQE                                                                        & 3 -- 11                                                 & \begin{tabular}[c]{@{}c@{}}0.0011\\ 1.0739\end{tabular} & 3 -- 7                                                 & 2    & \begin{tabular}[c]{@{}c@{}}8\\ 128\end{tabular}          & \begin{tabular}[c]{@{}c@{}}0.0104\\ 208.49\end{tabular} \\ \hline
QNN                                                                         & 2 -- 12                                                 & \begin{tabular}[c]{@{}c@{}}0.0006\\ 1.4682\end{tabular} & 2 -- 12                                                & 3    & \begin{tabular}[c]{@{}c@{}}4\\ 4096\end{tabular}         & \begin{tabular}[c]{@{}c@{}}0.0007\\ 1.6140\end{tabular} \\ \hline
QFTEntangled                                                                & 2 -- 9                                                  & \begin{tabular}[c]{@{}c@{}}0.0007\\ 14.668\end{tabular} & 2 -- 9                                                 & 3    & \begin{tabular}[c]{@{}c@{}}4\\ 512\end{tabular}          & \begin{tabular}[c]{@{}c@{}}0.0008\\ 14.822\end{tabular} \\ \hline
RAR                                                                         & 2 -- 12                                                 & \begin{tabular}[c]{@{}c@{}}0.0006\\ 1.4871\end{tabular} & 2 -- 12                                                & 3    & \begin{tabular}[c]{@{}c@{}}4\\ 4096\end{tabular}         & \begin{tabular}[c]{@{}c@{}}0.0007\\ 1.6308\end{tabular} \\ \hline
S2R                                                                         & 2 -- 12                                                 & \begin{tabular}[c]{@{}c@{}}0.0006\\ 1.4830\end{tabular} & 2 -- 12                                                & 3    & \begin{tabular}[c]{@{}c@{}}4\\ 4096\end{tabular}         & \begin{tabular}[c]{@{}c@{}}0.0007\\ 1.6196\end{tabular} \\ \hline
TLR                                                                         & 2 -- 12                                                 & \begin{tabular}[c]{@{}c@{}}0.0006\\ 1.4828\end{tabular} & 2 -- 12                                                & 3    & \begin{tabular}[c]{@{}c@{}}4\\ 4096\end{tabular}         & \begin{tabular}[c]{@{}c@{}}0.0007\\ 1.6085\end{tabular} \\ \hline
WState                                                                      & 2 -- 14                                                 & \begin{tabular}[c]{@{}c@{}}0.0004\\ 0.0087\end{tabular} & 2 -- 14                                                & 3    & \begin{tabular}[c]{@{}c@{}}2\\ 14\end{tabular}           & \begin{tabular}[c]{@{}c@{}}0.0008\\ 0.0098\end{tabular} \\ \hline
QASMBench                                                                   & 2 -- 13                                                 & \begin{tabular}[c]{@{}c@{}}0.0004\\ 0.0942\end{tabular} & 2 -- 13                                                & 1,2,3 & \begin{tabular}[c]{@{}c@{}}1\\ 128\end{tabular}          & \begin{tabular}[c]{@{}c@{}}0.0001\\ 205.57\end{tabular}
                              
\end{tabular}
\end{table}

\begin{tcolorbox}[size=title,rightrule=1mm, leftrule=1mm, toprule=0mm, bottomrule=0mm, arc=0pt,colback=gray!5,colframe=bleudefrance!75!black,breakable]
\textbf{Answer to RQ2:}  
StabQ effectively supports downstream quantum program analysis tasks through the unified Tableau Chain representation. Once the Tableau Chain is constructed, StabQ enables quantum state reconstruction and entanglement analysis over arbitrary execution steps without requiring repeated program execution or intermediate-state reconstruction. The state reconstruction process leverages $Zy=r~(\bmod~2)$ support analysis and relative phase reconstruction to accurately recover state information from symbolic representations. Moreover, the purity values obtained by StabQ for entanglement analysis are consistent with those computed from Schmidt decomposition over full density matrices, demonstrating that the Tableau Chain preserves the semantic information required for entanglement inference. These results demonstrate that Tableau Chain serves as a reusable symbolic intermediate representation for supporting systematic quantum program analysis.
\end{tcolorbox}

\subsection{Scalability with Respect to Circuit Characteristics}
\label{scalability}
Considering that different benchmark circuit families impose different upper bounds on the number of qubits for which a Tableau Chain model can be successfully constructed, we restrict our analysis to the common qubit range of 3 to 9, which is supported across all benchmarks. To reduce variance and enable a consistent cross-benchmark comparison, we aggregate all experimental results within this range. Furthermore, both gate count and Clifford rate are discretized into seven equal-width bins to facilitate a uniform analysis of their impact on construction overhead. The experimental results are summarized in ~\autoref{fig:3x2}.

The results demonstrate that StabQ exhibits good scalability in Tableau Chain model construction across all evaluated benchmarks. Both construction time and peak memory consumption increase with circuit size, while remaining within practical limits across the tested configurations.

As shown in Figures~\ref{fig:qubit-time} and ~\ref{fig:gate-time}, construction time increases steadily with both the number of qubits and the number of gates. This trend reflects the growing cost of tableau propagation and symbolic state updates during model construction. In contrast, Figures~\ref{fig:qubit-memory} and ~\ref{fig:gate-memory} indicate that peak memory consumption grows more moderately. Although larger circuits introduce greater variance and occasional high-memory outliers, most benchmarks remain within a relatively stable memory range, suggesting that the Tableau Chain representation maintains favorable memory efficiency.

The effect of the Clifford rate is less straightforward. Figures~\ref{fig:rate-time} and ~\ref{fig:rate-memory} show that neither construction time nor peak memory consumption exhibits a monotonic dependence on the Clifford rate. Instead, the overhead varies significantly across different Clifford-rate intervals, indicating that the performance of StabQ is influenced more by the overall symbolic structure of the circuit than by the proportion of Clifford gates alone.
The observed non-monotonic behavior suggests that the proposed symbolic representation and tableau consolidation strategy reduce the sensitivity of construction overhead to the proportion of non-Clifford operations. In some cases, circuits with higher Clifford rates even require less construction time, further supporting the idea that symbolic handling of non-Clifford gates, together with tableau consolidation, effectively mitigates their computational overhead during model construction.

Overall, the results suggest that StabQ’s construction cost is largely driven by circuit size, in terms of qubit and gate counts, while the influence of Clifford rate primarily reflects higher-order structural properties of the circuit rather than a direct linear effect.

\begin{figure}[!htbp]
    \centering

    \begin{subfigure}{0.48\textwidth}
        \centering
        \includegraphics[width=\textwidth]{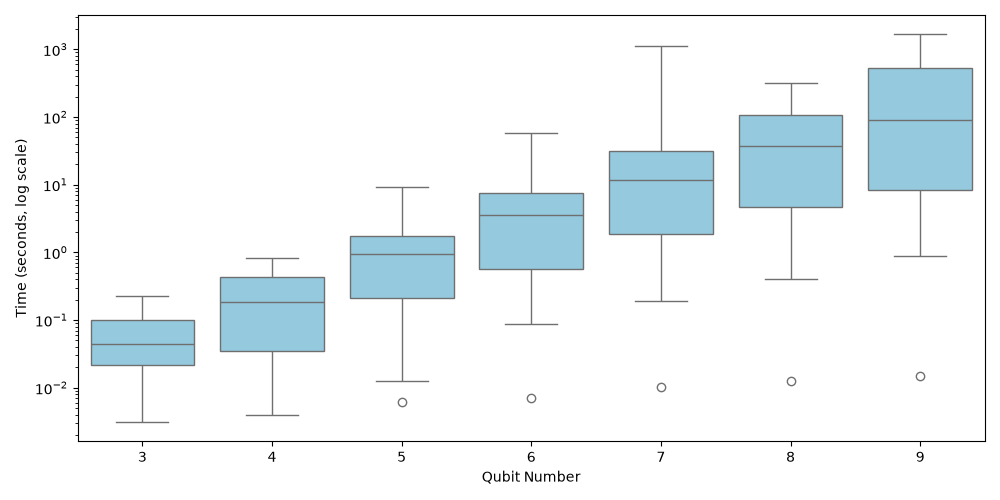}
        \caption{Distribution of the model building time of StabQ on benchmark circuits grouped by qubit number.}
        \label{fig:qubit-time}
    \end{subfigure}
    \hfill
    \begin{subfigure}{0.48\textwidth}
        \centering
        \includegraphics[width=\textwidth]{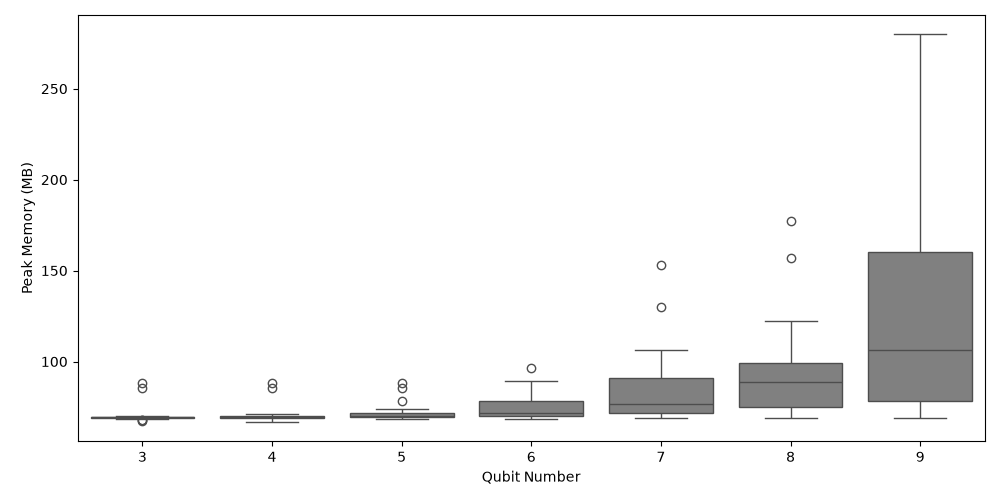}
        \caption{Distribution of the peak memory of StabQ on benchmark circuits grouped by qubit number.}
        \label{fig:qubit-memory}
    \end{subfigure}

    \vspace{0.5em}

    \begin{subfigure}{0.48\textwidth}
        \centering
        \includegraphics[width=\textwidth]{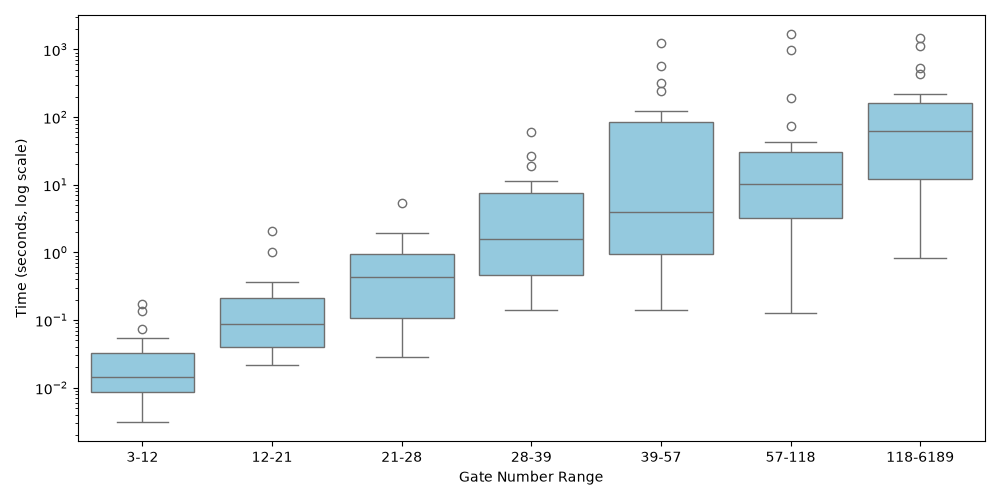}
        \caption{Distribution of the model building time of StabQ on benchmark circuits grouped by gate number.}
        \label{fig:gate-time}
    \end{subfigure}
    \hfill
    \begin{subfigure}{0.48\textwidth}
        \centering
        \includegraphics[width=\textwidth]{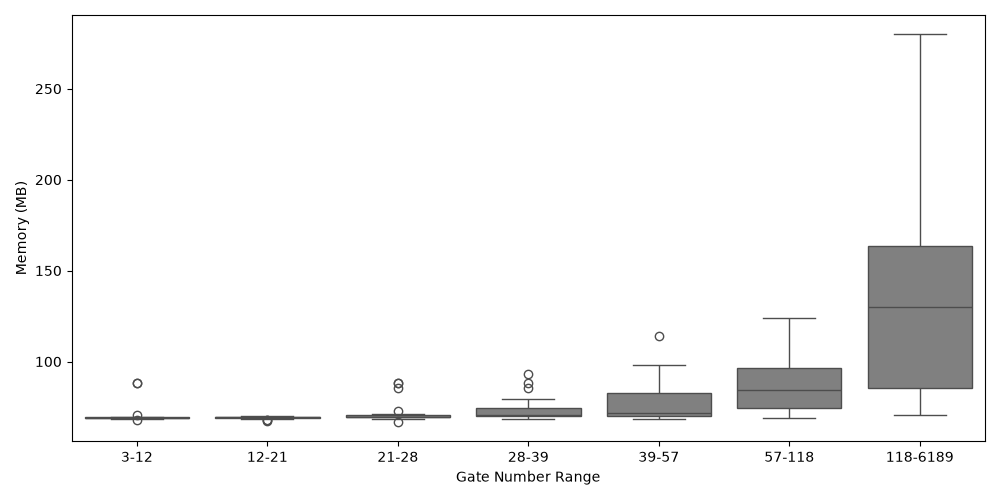}
        \caption{Distribution of the peak memory of StabQ on benchmark circuits grouped by gate number.}
        \label{fig:gate-memory}
    \end{subfigure}

    \vspace{0.5em}

    \begin{subfigure}{0.48\textwidth}
        \centering
        \includegraphics[width=\textwidth]{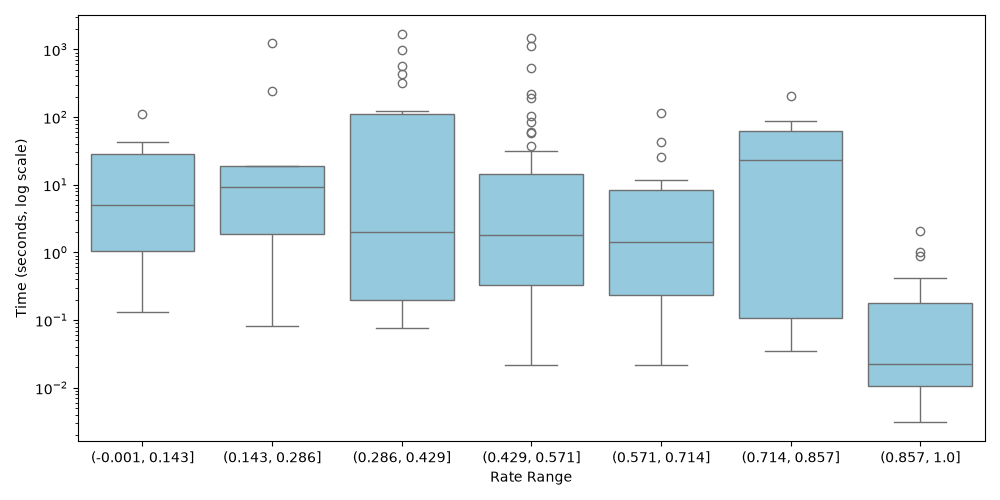}
        \caption{Distribution of the model building time of StabQ on benchmark circuits grouped by Clifford rate.}
        \label{fig:rate-time}
    \end{subfigure}
    \hfill
    \begin{subfigure}{0.48\textwidth}
        \centering
        \includegraphics[width=\textwidth]{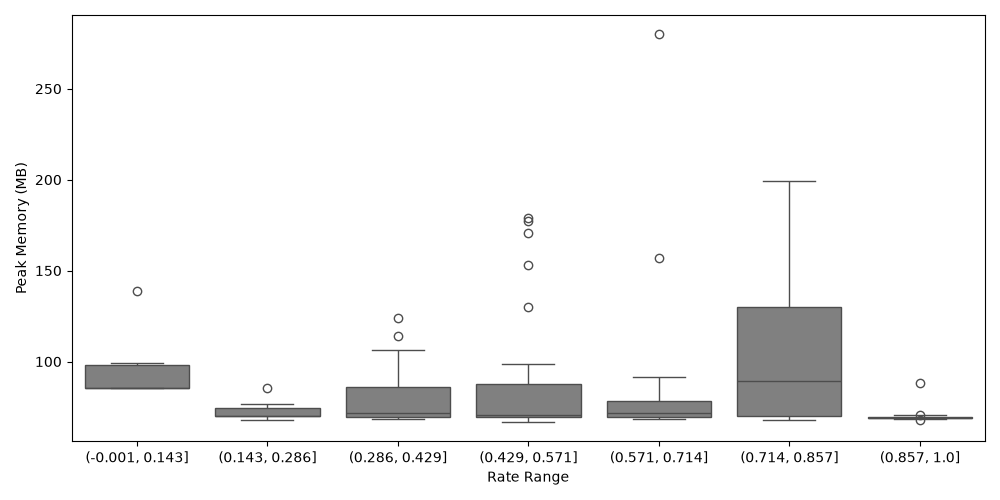}
        \caption{Distribution of the peak memory of StabQ on benchmark circuits grouped by Clifford rate.}
        \label{fig:rate-memory}
    \end{subfigure}

    \caption{Time and memory overhead of StabQ model construction across benchmark circuits. The box plots show the distributions of model building time and peak memory consumption when the benchmark circuits are grouped by qubit number, gate number, and Clifford rate. The y-axes of the time overhead are logarithmic.}
    \label{fig:3x2}
\end{figure}

\begin{tcolorbox}[size=title,rightrule=1mm, leftrule=1mm, toprule=0mm, bottomrule=0mm, arc=0pt,colback=gray!5,colframe=bleudefrance!75!black,breakable]
\textbf{Answer to RQ3:}  
StabQ demonstrates good scalability in Tableau Chain model construction across the evaluated benchmark circuits. The construction overhead mainly increases with circuit size, including the number of qubits and gates, while the peak memory consumption remains relatively stable within practical limits. The impact of Clifford rate is not directly proportional to the construction cost, suggesting that the overall symbolic structure of the circuit plays a more important role than the Clifford rate alone. These results indicate that the proposed symbolic representation and tableau consolidation strategy effectively control the growth of symbolic states and enable StabQ to construct scalable quantum program models.
\end{tcolorbox}

\section{Discussion}
\label{discussion}
The experimental results reveal several key factors influencing the overall scalability behavior of StabQ, along with its operational boundaries and potential directions for further improvement.

First, although StabQ leverages the stabilizer formalism for circuit evolution, the introduction of non-Clifford operations can still lead to an exponential growth in the number of tableaux, thereby imposing significant pressure on both construction time and memory consumption. While existing tableau consolidation techniques can partially alleviate this issue, an optimal consolidation strategy remains an open problem. Specifically, we formalize this challenge as follows:

\begin{problem}
Given a set of stabilizer groups $\{S_i\}_{i=1}^m$, where each $S_i$ stabilizes a quantum state $|x_i\rangle$, and a set of complex coefficients
$\boldsymbol{\alpha} = (\alpha_1, \alpha_2, \dots, \alpha_m)$ satisfying $\sum_{i=1}^m |\alpha_i|^2 = 1$,
we define the linear combination state
\[
|V\rangle = \sum_{i=1}^m \alpha_i |x_i\rangle.
\]
We aim to determine whether there exists a stabilizer group $S_V$ such that $|V\rangle$ is stabilized by $S_V$, i.e.,
\[
g |V\rangle = |V\rangle, \quad \forall g \in S_V,
\]
and if such a stabilizer group exists, how to construct $S_V$ from $\{S_i\}$ and $\boldsymbol{\alpha}$.
\end{problem}

This problem has been repeatedly encountered in the optimization of StabQ, and we illustrate it using two representative examples as follows:

\begin{example}
\[
\begin{aligned}
S_{1}=['+IX','+XI'],\ S_{2}=['-IX','+XI'],\ S_{3}=['+IX','-XI'],\ S_{4}=['-IX','-XI'] \\
\boldsymbol{\alpha} =[0.5, 0.5,0.5,-0.5]\ \ \ \ \ \ \  \Longrightarrow \ \ \ \ \ \ \   S_{V}=['+ZX','+XZ']
\end{aligned}
\]
\end{example}

\begin{example}
\[
\begin{aligned}
S_{1}=['+X'],\ S_{2}=['-X'],\ \boldsymbol{\alpha} =[0.5+0.5j,0.5-0.5j]\ \ \ \ \ \ \  \Longrightarrow \ \ \ \ \ \ \   S_{V}=['+Y']
\end{aligned}
\]
\end{example}

Second, we observe that the global phase recovery process may admit a more efficient formulation. The key observation is that stabilizer composition induces a finite symbolic space, where each Pauli string is restricted to ${I,X,Y,Z}$. Moreover, since non-Clifford operations are handled via Pauli decomposition, the evolution can ultimately be reduced to a finite set of Clifford operations. This suggests that a pattern-matching-based strategy may improve efficiency. However, ensuring soundness of such an approach remains challenging, and therefore it is not adopted in the current framework.

Third, we consider Case 2 of the entanglement analysis (Section~\ref{entangle}). Although this setting exhibits a tensor-product-like structure within each tableau branch, our experiments show that the current processing strategy is less efficient than that used for Case 3. This indicates that the handling of Case 2 is not yet optimal in the current framework, and improving its processing strategy remains an important direction for future work.

Overall, the above observations suggest that the scalability of StabQ is primarily governed by structural properties of quantum circuits rather than individual gate-level characteristics. In particular, the interaction between non-Clifford operations and tableau evolution plays a central role in determining computational overhead, while global structural features such as symbolic redundancy and entanglement organization further influence performance. These findings highlight that improving StabQ requires not only optimizing low-level operations such as tableau consolidation and phase recovery, but also developing more structure-aware strategies that better capture higher-level circuit characteristics.

\section{Threats to Validity}
\label{validity}
We discuss the potential threats to the validity of our study from four perspectives.

\textbf{Internal Validity.}
The correctness of StabQ depends on the implementation of several core components, including stabilizer propagation, Pauli decomposition, global-phase recovery, and tableau consolidation. Errors in these components may affect the correctness of the constructed Tableau Chain and downstream analysis results. To mitigate this threat, we validate the correctness of the constructed symbolic models by reconstructing quantum states from Tableau Chains and comparing them with the exact statevectors generated by a Qiskit statevector simulator.

\textbf{External Validity.} The generalizability of our evaluation may be affected by the characteristics and scale of the selected benchmark programs. Although our experiments include multiple benchmark suites with diverse circuit structures and Clifford rates, the evaluated circuits remain limited in scale due to the complexity of symbolic representations that involve non-Clifford operations. Moreover, StabQ is currently implemented as a research prototype and does not fully exploit all possible engineering optimizations. Therefore, the reported performance reflects the effectiveness of the proposed symbolic execution approach rather than the maximum efficiency achievable by a fully optimized implementation. Future work will investigate advanced symbolic reduction techniques and evaluate StabQ on larger-scale quantum applications. 

\textbf{Construct Validity.} We evaluate StabQ based on symbolic model correctness, downstream analysis capability, and scalability. However, these metrics may not fully capture all aspects of a symbolic execution framework. For example, the number of tableaux is used as an indicator of growth in symbolic representation, but it does not fully characterize memory consumption or computational complexity. In addition, our evaluation of downstream analysis focuses on quantum state reconstruction and entanglement analysis, while other potential quantum program analysis tasks remain unexplored. 

\textbf{Conclusion Validity.} Our conclusions are derived from the experimental results obtained on the evaluated benchmark programs. Although the correctness of reconstructed states and entanglement properties is verified against exact references, the observed performance trends may vary for different circuit structures, qubit scales, and Clifford rates. Therefore, the conclusions should be interpreted within the scope of the evaluated benchmarks.

\section{Related Work}
\label{related}
This work is closely related to three major lines of research: static analysis techniques for quantum programs, specialized semantic models for quantum program analysis, and applications of the stabilizer formalism.

\subsection{Static Analysis of Quantum Programs}
A substantial body of work applies classical static analysis techniques to quantum programs. These efforts have primarily focused on two fundamental analysis objectives: support analysis, which determines the qubits involved in quantum computation, and entanglement analysis, which infers entanglement relationships among qubits.
Yu and Palsberg~\cite{yu2021quantum} proposed a quantum abstract interpretation based on projection abstractions, representing quantum states as tuples of low-dimensional projections rather than full density matrices. This abstraction enables efficient state propagation while preserving only support information for subsequent analysis.
Perdrix~\cite{perdrix2008entanglement} pioneered an abstract interpretation that over-approximates entanglement relations among qubits as a partition of the quantum system. Honda~\cite{honda2015analysis} refined this direction by tracking abstract stabilizer information to approximate entanglement more precisely. 
Xia and Zhao~\cite{static} proposed a static analysis approach that constructs an entanglement connectivity graph directly from quantum program code, enabling the inference of entanglement relationships among qubits throughout quantum state evolution.
Type-system-based approaches such as Twist~\cite{yuan2022twist} combine static reasoning with runtime checks to enforce purity of quantum data, while linters and bug detectors such as QChecker~\cite{zhao2023qchecker} and LintQ~\cite{paltenghi2024lintq} detect common bug patterns in real-world quantum programs through abstraction- and pattern-based static checks.

These static analysis approaches provide efficient analysis with formal soundness guarantees in many cases. However, due to the use of abstraction, they may sacrifice precision and introduce false positives, limiting their applicability to fine-grained quantum program analysis. Moreover, existing approaches are typically designed for specific analysis objectives and rely on task-specific semantic models. Developing a unified representation that supports diverse quantum program analysis tasks remains challenging, limiting its generality and extensibility.

\subsection{Specialized Models for Quantum Program Analysis}
Another line of work builds dedicated semantic models for particular analysis tasks. 
Decision-diagram-based representations~\cite{niemann2016qmdds,BDD,feynmanDD,burgholzer2021advanced} have been widely explored for quantum circuit simulation and equivalence checking due to their ability to compactly encode structured quantum states and operations. 
The ZX-calculus~\cite{ZX1,coecke2011interacting,kissinger2020pyzx} provides a graphical rewriting framework that represents quantum circuits as algebraic graphs and applies sound rewrite rules for circuit optimization and equivalence checking. 
Tensor-network-based approaches~\cite{markov2008simulating} exploit structural properties of quantum circuits, such as limited connectivity and small treewidth, to enable efficient classical simulation through tensor contraction.
Quantum Markov chain-based approaches~\cite{feng2015qpmc,ying2021model,Qreach}model quantum programs as state-transition systems and enable reachability analysis and model checking of quantum behaviors.
Tree-automata-based approaches~\cite{tree1, tree2} provide symbolic representations of sets of quantum states and support automated verification of quantum circuits through automata-based reasoning.

Although these semantic models provide powerful abstractions for quantum program analysis, many of them are adapted from general-purpose computational models or mathematical frameworks originally developed outside quantum computing. Consequently, applying these models to quantum programs often requires additional abstractions, constraints, or domain-specific adaptations to capture quantum-specific properties. In contrast, the stabilizer formalism originates directly from quantum information theory and exploits the algebraic structure of quantum states and operations, providing inherent advantages for analyzing quantum-specific properties.

\subsection{Stabilizer Formalism and Its Applications}
The stabilizer formalism~\cite{gottesman1997}, originating in quantum error-correction theory, has become a fundamental framework for constructing and analyzing quantum error-correcting codes and fault-tolerant quantum computation~\cite{steane,shor,surface}. Beyond its original applications, the stabilizer formalism has also emerged as an efficient representation for quantum states and operations, particularly in the context of Clifford circuits. Gottesman and Knill's theorem demonstrates that Clifford circuits can be efficiently simulated classically by tracking stabilizer generators rather than explicitly manipulating exponentially large state vectors. Based on this principle, various stabilizer-based simulators, such as the Aaronson-Gottesman simulator~\cite{aaronson2004improved} and Stim~\cite{gidney2021stim}, have been developed to achieve scalable simulation of large-scale Clifford circuits.

Due to its compact representation and algebraic structure, the stabilizer formalism has also been explored for quantum program analysis and verification tasks. In particular, stabilizer-based approaches have been extensively studied for quantum entanglement analysis, where the structure of stabilizer generators is leveraged to characterize entanglement relationships among qubits~\cite{fattal,StabE1,StabE2,StabE3,honda2015analysis}. Beyond quantum state analysis, Fang and Ying~\cite{stabqec} integrated stabilizer representations into existing quantum program symbolic execution techniques~\cite{bauer2023symqv} and applied their approach to quantum error-correction programs, thereby significantly improving the efficiency of formal verification. Tan et al.~\cite{hornbro} proposed HornBro, a quantum program repair framework that combines stabilizer-based representations with the Z3 solver for fault localization and automated repair. These studies demonstrate that stabilizer representations provide compact and structured abstractions for analyzing quantum programs.

However, existing stabilizer-based applications remain fundamentally constrained by the Clifford restriction. Non-Clifford operations violate the closure property of stabilizer representations, preventing conventional stabilizer frameworks from directly preserving quantum state evolution in general quantum programs. Therefore, extending stabilizer-based representations to support non-Clifford operations while preserving their efficiency remains an important challenge.

Several studies~\cite{Bravyi2016,Bravyi2019} have attempted to address this limitation through stabilizer decomposition techniques. These approaches introduce auxiliary magic states and represent non-stabilizer states as linear combinations of stabilizer states, thereby enabling the simulation of quantum circuits containing both Clifford and non-Clifford operations, such as Clifford+T circuits. By decomposing non-stabilizer components into stabilizer-state representations, they effectively extend stabilizer-based simulation beyond the Clifford-only setting.

Nevertheless, existing stabilizer decomposition approaches primarily target efficient quantum simulation rather than symbolic analysis of quantum programs. To achieve scalability, many approaches reduce simulation overhead through approximation techniques, such as truncating low-weight stabilizer components, which may introduce errors into intermediate state representations. As a result, these approaches are not directly suitable for quantum program analysis tasks that require faithful preservation of execution semantics and precise characterization of intermediate quantum states. Moreover, stabilizer decomposition methods typically focus on representing non-stabilizer states during simulation, rather than maintaining a symbolic representation of the entire program execution process.

Different from these approaches, StabQ focuses on constructing an exact symbolic representation of quantum program evolution based on stabilizer representations. Instead of approximating non-stabilizer states or decomposing states after they are generated, StabQ transforms non-Clifford operations during symbolic propagation and explicitly maintains their effects throughout program execution. This design enables stabilizer-based symbolic execution of general quantum programs that contain both Clifford and non-Clifford operations while preserving the semantic information required for downstream quantum program analysis.

\section{Conclusion}
\label{conclusion}
In this paper, we presented \textbf{StabQ}, a symbolic execution framework for quantum program analysis based on stabilizer representations. StabQ extends the applicability of stabilizer-based representations beyond Clifford-only programs by enabling symbolic state propagation in quantum programs that contain non-Clifford operations. Through a Pauli-decomposition mechanism, StabQ transforms non-Clifford operations into weighted Pauli components and integrates them into a unified symbolic execution process based on weighted stabilizer tableaux.
Based on this representation, StabQ constructs a \textit{Tableau Chain} that captures the evolution of quantum program states throughout execution. Each Tableau Chain node records an intermediate symbolic state, allowing multiple analysis tasks to be performed over different execution stages without repeatedly reconstructing quantum states. Building upon the Tableau Chain, StabQ supports various quantum program analysis tasks, including quantum state reconstruction, entanglement analysis, support-set computation, and Clifford-property detection.
To improve the scalability and maintainability of symbolic representations, StabQ incorporates tableau consolidation and global-phase recovery mechanisms. These techniques reduce redundant symbolic states and preserve the semantic information required for accurate downstream analysis.
Our experimental evaluation demonstrates that StabQ constructs semantically consistent symbolic execution models and faithfully preserves quantum program behavior. Across diverse benchmark suites and circuit characteristics, the reconstructed quantum states and entanglement properties are consistent with exact reference results, demonstrating the correctness of the proposed representation. Furthermore, the results show that the Tableau Chain provides a reusable intermediate representation that effectively supports downstream quantum program analysis.
Our study also identifies several limitations of the current framework. In particular, symbolic representation growth remains challenging for programs with intensive non-Clifford operations, global-phase recovery introduces additional computational overhead, and certain entanglement structures require more specialized analysis strategies. These observations motivate future research on structure-aware symbolic representations, more efficient tableau consolidation techniques, and advanced optimization strategies for large-scale quantum programs.

Overall, StabQ provides a practical foundation for symbolic execution of quantum programs by extending stabilizer representations to support general quantum computations. By combining the compactness of stabilizer representations with explicit execution-state tracking, StabQ bridges the gap between efficient quantum state representation and systematic quantum program analysis in the presence of both Clifford and non-Clifford operations.

\section{Data Availability}
Our framework is publicly available at \url{https://github.com/Xzore19/StabQ}.

\begin{acks}
This work was supported by JST SPRING Grant No.\ JPMJSP2136, JST Moonshot R\&D Grant No.\ JPMJMS256E and JSPS KAKENHI Grants JP24K14908, JP26K02892, JP26K23820.
\end{acks}

\bibliographystyle{ACM-Reference-Format}
\bibliography{paper/ref}










\end{document}